\documentclass[showpacs,preprintnumbers,amsmath,amssymb,oneside,10pt]{article}
\usepackage{amsfonts,amsmath,amssymb,mathrsfs}
\usepackage{dsfont}
\usepackage{graphicx}

\newcommand{\bdm}{\begin{displaymath}}
\newcommand{\edm}{\end{displaymath}}
\newcommand{\bea}{\begin{eqnarray}}
\newcommand{\eea}{\end{eqnarray}}
\newcommand{\beas}{\begin{eqnarray*}}
\newcommand{\eeas}{\end{eqnarray*}}
\newcommand{\bay}{\begin{array}{c}}
\newcommand{\eay}{\end{array}}
\newcommand{\ben}{\begin{enumerate}}
\newcommand{\een}{\end{enumerate}}
\newcommand{\be}{\begin{equation}}
\newcommand{\ee}{\end{equation}}

\newcommand{\LZ}{L^2(\mathbb{R}^3\to\mathbb{C})}
\newcommand{\LZN}{L^2(\mathbb{R}^{3N}\to\mathbb{C})}

\newtheorem{theorem}{Theorem}[section]
\newtheorem{lemma}[theorem]{Lemma}
\newtheorem{notation}[theorem]{Notation}

\newtheorem{corollary}[theorem]  {Corollary}
\newtheorem{remark}[theorem]  {Remark}
\newtheorem{definition}[theorem] {Definition}
\newtheorem{proposition}[theorem]{Proposition}

\newenvironment{proof}{\emph{Proof:}}{\begin{flushright} $ \Box $ \end{flushright}}

\renewcommand{\phi}{\varphi}

\begin{document}

\title{On the Time Dependent Gross Pitaevskii- and Hartree Equation}

\author{P. Pickl
\footnote{ Mathematisches Institut der Universit\"at M\"unchen,
Theresienstr. 39, 80333 M\"unchen,  E-mail:
pickl@mathematik.uni-muenchen.de}}

\maketitle

\begin{abstract}
We are interested in solutions $\Psi_t$ of the Schr\"odinger
equation of $N$ interacting bosons under the influence of a time
dependent external field, where the range and the coupling constant
of the interaction scale with $N$ in such a way, that the
interaction energy per particle stays more or less constant. Let
$\mathcal{N}^{\phi_0}$ be the particle number operator with respect
to some $\phi_0\in L^2(\mathbb{R}^3\to\mathbb{C})$. Assume that the
relative particle number of the initial wave function $N^{-1}\langle
\Psi_0,\mathcal{N}^{\phi_0}\Psi_0\rangle$ converges to one as
$N\to\infty$. We shall show that we can find a $\phi_t\in
L^2(\mathbb{R}^3\to\mathbb{C})$ such that
$\lim_{N\to\infty}N^{-1}\langle
\Psi_t,\mathcal{N}^{\phi_t}\Psi_t\rangle=1$ and that $\phi_t$ is
--- dependent of the scaling of the range of the interaction --- solution of the Gross-Pitaevskii or Hartree equation.

We shall also show that under additional decay conditions of
$\phi_t$ the limit can be taken uniform in $t<\infty$ and that
convergence of the relative particle number implies convergence of
the $k$-particle density matrices of $\Psi_t$.

\end{abstract}

\section{Introduction}

In this paper we wish to analyze the dynamics of a Bose condensate
of $N$ interacting particles when the external trap --- described by
an external potential $A_t$ --- is changed, for example removed.

We are interested in solutions of the $N$-particle Schr\"odinger
equation \be\label{schroe} i\frac{d}{dt} \Psi_t = H\Psi_t \ee with
some symmetric $\Psi_0$ we shall specify below and the Hamiltonian
\be\label{hamiltonian}
 H=-\sum_{j=1}^N \Delta_j+\sum_{j\neq k=1}^n v^N_{\beta}(x_j-x_k) +\sum _{j=1}^N A_t(x_j)
 \ee
acting on the Hilbert space $\LZN$, where $\beta\in\mathbb{R}$
stands for the scaling behavior of the interaction. The
$v^N_{\beta}$ we wish to analyze scales with the particle number in
such a way, that the total interaction energy scales in the same way
as the total kinetic energy of the $N$ particles. For the heuristic
arguments we shall give first one should think of a interaction
which is given by $v^N_{\beta}(x)= N^{-1+3\beta} v(N^\beta x)$ for a
compactly supported, spherically symmetric, positive potential $v\in
L^\infty$. The interactions we shall choose below will be of a more
general form. The $A_t$ describing the trap potential is a time
dependent external potential which we shall choose
--- in contrast to $v^N_{\beta}$
--- not $N$-dependent. Note, that $H$ conserves symmetry, i.e. for
any symmetric function $\Psi_0$ also $H\Psi_0$ and thus $\Psi_t$ is
symmetric.

Assume moreover that the initial wave function $\Psi_0$ is a
condensate, i.e. that there exists a $L^2$ function $\phi_0$ such
that
$$\lim_{N\to\infty}\langle \Psi_0,\widehat{n}^{\phi_0}\Psi_0\rangle=1$$
where $\widehat{n}^{\phi_0}$ is the particle number operator of
particles in $\phi_0$ (see Definition \ref{defpro} (c) and Lemma
\ref{kombinatorik} (a)).

 Under these and some additional technical assumptions
we shall show that also $\Psi_t$ will be a condensate, i.e.  that
there exist $L^2$ functions $\phi_t$ such that
$$\lim_{N\to\infty}\langle
\Psi_0,\widehat{n}^{\phi_t}\Psi_0\rangle=1$$ uniform in $t$ on any
compact subset of $\mathbb{R}^+$ and --- under additional decay
conditions on $\phi_t$
--- uniform in $t\in \mathbb{R}^+$.

Even more: We shall show that $\phi_t$ solves the differential
equation \be\label{meanfield}i\frac{d}{dt}
\phi_t=-\left(\Delta+A_t+V_{\phi_t}\right)\phi_t\ee with $\phi_0$ as
above, where the ``mean field'' $V_{\phi_t}$ depends on $\phi_t$
itself, so (\ref{meanfield}) is a non-linear equation. For different
regimes of $\beta$ different effective mean field potentials will
appear:
\begin{center}
\begin{tabular}{|c|c|c|c|}
  \hline
$\beta<0$  & $\beta=0$ & $0<\beta\leq1$ & $\beta>1$   \\\hline
 $V_{\phi_t}=0$ & $V_{\phi_t}=v\star |\phi_t|^2$ & $V_{\phi_t}=2a|\phi_t|^2$ &  $V_{\phi_t}=0$ \\
  \hline
\end{tabular}
\end{center}
We explain the table. For $\beta<0$
$\lim_{N\to\infty}\left\|\sum_{j=2}^N
v^N_{\beta}(x_1-x_j)\right\|_\infty=0$, so it is heuristically clear
that the mean field is zero.

$v^N_{0}=v/N$ and hence particle (say number one) feels
$N^{-1}\sum_{j=2}^Nv(x_1-x_j)\approx \int v(x-y)|\phi_t|^2(y)d^3y$
assuming that the particles are $|\phi_t|^2$-distributed. In this
case (\ref{meanfield}) is called ``Hartree equation''. This limit
has already been proven in the literature \cite{froehlich}. A sketch
of an alternative proof shall be given below to motivate the
technique used in this paper for the cases $0<\beta\leq 1$ we shall
focus on here.

For $0<\beta$ the interaction becomes $\delta$-like. To be able to
``average out'' the potential it is important to control the
microscopic structure of $\Psi_t$. Assuming that the energy of
$\Psi_t$ is small, the microscopic structure is --- whenever two
particles approach  --- roughly given by the zero energy scattering
length of the potential $1/2 v^N_{\beta}$ (the factor $1/2$ comes
from the fact that one has to go to relative coordinates of the two
particles).

For $\beta=1$ the scaling of the potential is such that the zero
energy scattering state of $f^N(x)$ of the potential $v^N_{\beta}/2$
scales like $f^N=f_1(Nx)$. It follows that the mean field is given
by $2a|\phi_t|^2$, where $a$ is the scattering length of $v/2$.

For $0<\beta<1$ the scaling is ``softer'' and the microscopic
structure disappears as $N\to\infty$. Thus the mean field is given
by $V_{\phi_t}=\|v\|_1|\phi_t|^2$. One can also argue, that for
``soft scalings'' the scattering length is in good approximation
given by the first order Born approximation, i.e. by the $L_1$-norm
of the interaction.

For $\beta>1$ note, that the scattering length of a spherically
symmetric potential is always smaller than the radius of its
support, thus for $\beta>1$ $Na_N\to 0$ for $N\to \infty$, implying
that the interaction becomes negligible for $\beta>1$ as
$N\to\infty$.

The cases $\beta=1$ and $0<\beta<1$ have been proven recently for
the special case $A_t\equiv0$ \cite{erdos1,erdos2,erdos3,erdos4}. We
shall give an alternative proof including time dependent external
potentials and with weaker conditions on $\Psi_0$ and also
generalizing to hard core potentials for $\beta=1$.

\section{Definition of the Projectors}

Before we consider the different cases of $0\leq\beta\leq1$ we
define the following operators acting on $\LZN$ we shall need in the
proofs below
\begin{definition}\label{defpro}
For any $\phi\in\LZ$ we define
\begin{enumerate}
\item for any $1\leq j\leq N$ and any $\phi\in\LZ$ the orthogonal projector $p_j^{\phi}$ of the $j^{\text{th}}$ particle onto
$\phi$  defined by
$$p_j^{\phi}\Psi:=\phi(x_j)\int\phi(x_j)^*\Psi(x_1,\ldots,x_N) d^3 x_j$$ for
any $\Psi\in\LZN$. We shall also need $q_j^{\phi}=1-p_j^{\phi}$.

\item For any $0\leq k\leq j\leq N$ we define the set $$\mathcal{A}_k^j:=\{(a_1,a_2,\ldots,a_j): a_l\in\{0,1\}\;;\;
\sum_{l=1}^j a_l=k\}\;.$$ For any $0\leq k\leq j\leq N$
 and any $\phi(x_j)\in\LZ$ we define the
orthogonal projector $P_{j,k}^\phi$ acting on $\LZN$ as
$$P_{j,k}^\phi:=\sum_{a\in\mathcal{A}_k^j}\prod_{l=1}^j\big(p_{N-j+l}^{\phi}\big)^{1-a_l} \big(q_{N-j+l}^{\phi}\big)^{a_l}$$
and denote the special case $j=N$ by $P_k^\phi:=P_{N,k}^\phi$. For
negative $k$ and $k>N$ we set $P_k^\phi:=0$.
\item
For any function $f:\{0,1,\ldots,N\}\to\mathbb{R}^+$ we define the
operator $\widehat{f}^{\phi}:\LZN\to\LZN$ as
\be\label{hut}\widehat{f}^{\phi}:=\sum_{j=0}^{N} f(j)P_j^\phi\;.\ee

We shall also need translations of the operators $\widehat{f}$: Let
$f:\{0,1,\ldots,N\}\to\mathbb{R}^+$ and  $d\in\mathbb{Z}$. We define
the operator $\widehat{f}^{\phi}_d:\LZN\to\LZN$ as
$$\widehat{f}^{\phi}_d:=\sum_{j=d}^{N+d} f(j-d)P_j^\phi\;.$$
\end{enumerate}
\end{definition}
\begin{notation}
Throughout the paper hats $\;\widehat{\cdot}\;$ shall solemnly be
used in the sense of Definition \ref{defpro} (c). In what follows
the letter $C$ will be used for various constants that need
 not be identical even within the same equation. 
\end{notation}
With Definition \ref{defpro} we arrive directly at the following Lemma
based on combinatorics of the $p_j^\phi$ and $q_j^\phi$:
\begin{lemma}\label{kombinatorik}
\begin{enumerate}
\item For any functions $f,g\{0,1,\ldots,N\}\to\mathbb{R}^+$ we have
that
$$\widehat{f}\widehat{g}=\widehat{fg}=\widehat{g}\widehat{f}\;\;\;\;\;\;\;\;\;\;\widehat{f}p_j=p_j\widehat{f}\;\;\;\;\;\;\;\;\;\;\widehat{f}P_{j,k}=P_{j,k}\widehat{f}\;.$$

\item Let $n:\{0,1,\ldots,N\}\to\mathbb{R}^+$ be given by $n(k):=\sqrt{k/N}$.
Then the respective $(\widehat{n}^{\phi})^2$ (c.f. (\ref{hut}))
equals the relative particle number operator of particles not in the
state $\phi$, i.e.
$$(\widehat{n}^{\phi})^2=N^{-1}\sum_{j=1}^Nq_j^\phi\;.$$

\item For any function $f:\{0,1,\ldots,N\}\to\mathbb{R}^+$ and any symmetric $\Psi\in\LZN$ we have \bea\label{komb1}
\left\| \widehat{f}^\phi q^\phi_1\Psi\right\|^2&=&
\|\widehat{f}^\phi\widehat{n}^{\phi}\Psi\|^2\\
\label{komb2} \left\| \widehat{f}^\phi
q^\phi_1q^\phi_2\Psi\right\|^2&\leq&
\frac{N}{N-1}\|\widehat{f}^\phi(\widehat{n}^{\phi})^2\Psi\|^2\eea

\item For any function $f:\{0,1,\ldots,N\}\to\mathbb{R}^+$, any function $v:\mathbb{R}^6\to\mathbb{R}$ and any
$j,k=0,1,2$ we have $$\widehat{f}^\phi Q^\phi_j v(x_1,x_2)Q^\phi_k=
Q^\phi_j v(x_1,x_2)\widehat{f}^\phi_{k-j}Q^\phi_k\;,$$ where
$Q^\phi_0:=p^\phi_1 p^\phi_2$, $Q^\phi_1:=p^\phi_1q^\phi_2$ and
$Q^\phi_2:=q^\phi_1q^\phi_2$.

\item For any $w\in L^\infty(\mathbb{R}^3\to\mathbb{C})$ and any symmetric $\Psi\in\LZN$
\be\label{komb6} \left|\langle
\Psi,w(x_1)\Psi\rangle-\langle\phi,w\phi\rangle\right|\leq
4\|w\|_\infty
\left(N^{-1/4}+\|(\widehat{n}^{\phi})^{1/2}\Psi\|^2\right)\;.
 \ee
\end{enumerate}
\end{lemma}

\begin{proof}
(a) follows immediate from definition \ref{defpro}, using that $p_j$
and $q_j$ are orthogonal projectors.

For (b) note that $1=\sum_{k=1}^N P_k^\phi$. Using also
$(q_k^\phi)^2=q_k^\phi$ and $q_k^\phi p_k^\phi=0$ we get \beas
N^{-1}\sum_{k=1}^Nq_k^\phi&=&N^{-1}\sum_{k=1}^Nq_k^\phi\sum_{j=1}^N
P_j^\phi= N^{-1}\sum_{j=1}^N\sum_{k=1}^Nq_k^\phi
P_j^\phi=N^{-1}\sum_{j=1}^Nj P_j^\phi\eeas and (b) follows.

For (\ref{komb1}) we can write using symmetry of $\Psi$
\beas\|\widehat{f}^\phi\widehat{n}^{\phi}\Psi\|^2
&=&\langle\Psi,(\widehat{f}^\phi)^{2}(\widehat{n}^{\phi})^2\Psi\rangle=N^{-j}\sum_{k=1}^N\langle\Psi,(\widehat{f}^\phi)^{2}q_k^\phi\Psi\rangle
\\&=&\langle\Psi,(\widehat{f}^\phi)^{2}q_1^\phi\Psi\rangle
=\langle\Psi,q_1^\phi(\widehat{f}^\phi)^{2}q_1^\phi\Psi\rangle=\|(\widehat{f}^\phi)
q^{\phi}_1\Psi\|^2\;.
 \eeas
Similarly we have for (\ref{komb2}) \beas \|\widehat{f}^\phi
(\widehat{n}^{\phi})^2\Psi\|^2 &=&\langle\Psi,(\widehat{f}^\phi
)^2(\widehat{n}^{\phi})^2\Psi\rangle
=N^{-2}\sum_{j,k=1}^N\langle\Psi,(\widehat{f}^\phi )^2q_j^\phi
q_k^\phi\Psi\rangle
\\&=&\frac{N-1}{N}\langle\Psi,(\widehat{f}^\phi )^2q_1^\phi
q_2^\phi\Psi\rangle+N^{-1}\langle\Psi,(\widehat{f}^\phi
)^2q_1^\phi\Psi\rangle
\\&=&\frac{N-1}{N} \|\widehat{f}^\phi q^{\phi}_1q^{\phi}_2\Psi\|+N^{-1}\|\widehat{f}^\phi
q^{\phi}_1\Psi\|
 \eeas
and (c) follows.

Using the definitions above we have for (d) \beas \widehat{f}^\phi
Q^\phi_j v(x_1,x_2)Q^\phi_k
&=&\sum_{l=0}^N f(l)P_l^\phi Q^\phi_jv(x_1,x_2)Q^\phi_k
\\= \sum_{l=0}^N f(l)P_{N-2,l-j}^\phi Q^\phi_jv(x_1,x_2)Q^\phi_k
&=& \sum_{l=k-j}^{N+k-j}  Q^\phi_jv(x_1,x_2)f(l+j-k)P_{N-2,l-k}^\phi
Q^\phi_k
\\= \sum_{l=k-j}^{N+k-j}  Q^\phi_jv(x_1,x_2)f(l+j-k)P_{l}^\phi Q^\phi_k
&=&Q^\phi_j v(x_1,x_2)\widehat{f}^\phi _{k-j}Q^\phi_k
 \eeas

For (e) we have \beas \left|\langle
\Psi,w(x_1)\Psi\rangle-\langle\phi,w\phi\rangle\right|
&=&\big|\langle p^\phi_1\Psi,w(x_1)p^\phi_1\Psi\rangle+\langle
p^\phi_1\Psi,w(x_1)q^\phi_1\Psi\rangle+\langle
q^\phi_1\Psi,w(x_1)p^\phi_1\Psi\rangle\\&&+\langle
q^\phi_1\Psi,w(x_1)q^\phi_1\Psi\rangle-\langle\phi,w\phi\rangle\big|
\\&\leq&\langle\phi,w\phi\rangle
\left(1-\|p^\phi_1\Psi\|^2\right)+\|w\|_\infty
\|q^\phi_1\Psi\|^2+2|\langle q^\phi_1\Psi,w(x_1)p^\phi_1\Psi\rangle|
\\&\leq&2\|w\|_\infty\|q^\phi_1\Psi\|^2+2\|(\widehat{n}^\phi)^{-1/2}q^\phi_1\Psi\|\;\|(\widehat{n}^\phi)^{1/2}w(x_1)p^\phi_1\Psi\|\;.
 \eeas
Using that $\sqrt{k+1}<\sqrt{k}+1$ (thus $n(k+1)<n(k)+N^{-1/2}$) and
part (d)  \beas
\|(\widehat{n}^\phi)^{1/2}w(x_1)p^\phi_1\Psi\|^2&=&\langle
\Psi,p^\phi_1w(x_1)\widehat{n}^\phi q^\phi_1
w(x_1)p^\phi_1\Psi\rangle+\langle
\Psi,p^\phi_1w(x_1)\widehat{n}^\phi p^\phi_1
w(x_1)p^\phi_1\Psi\rangle
\\&=&\langle
(\widehat{n}^\phi)^{1/2}\Psi,p^\phi_1w(x_1) q^\phi_1
w(x_1)p^\phi_1(\widehat{n}^\phi)^{1/2}\Psi\rangle\\&&+\langle
(\widehat{n}^\phi_{-1})^{1/2}\Psi,p^\phi_1w(x_1)p^\phi_1
w(x_1)p^\phi_1(\widehat{n}^\phi_{-1})^{1/2}\Psi\rangle
\\&\leq&\|w(x_1)\|_\infty^2\left(2\|(\widehat{n}^\phi)^{1/2}\Psi\|^2+N^{-1/2}\right)\eeas
 thus part (c) of the Lemma yields
$$\left|\langle
\Psi,w(x_1)\Psi\rangle-\langle\phi,w\phi\rangle\right|\leq
4\|w\|_\infty
\left(\|\widehat{n}^{\phi}\Psi\|^2+N^{-1/4}+\|(\widehat{n}^{\phi})^{1/2}\Psi\|^2\right)$$
With the operator inequality
$(\widehat{n}^{\phi})^\lambda<(\widehat{n}^{\phi})^\gamma$ for any
$\lambda<\gamma$ we get (e).
\end{proof}

\subsection{Convergence of the Reduced Density Matrix}

\begin{lemma}\label{kondensat}
Let $j>0$, $\phi\in L^2$ and let $\Psi\in\LZN$ be symmetric, let $\mu(\Psi)$ be the reduced one particle density matrix of $\Psi$. Then
\begin{enumerate}\item
$$\lim_{N\to\infty}\|\widehat{n}^\phi\Psi\|=0
\Leftrightarrow
\lim_{N\to\infty}\left\langle\Psi,\left(\widehat{n}^\phi\right)^j\Psi\right\rangle=0
\;.$$

\item $$\lim_{N\to\infty}\left\langle\Psi,\left(\widehat{n}^\phi\right)^j\Psi\right\rangle=0
\Rightarrow
\lim_{N\to\infty}\mu(\Psi)=|\phi\rangle\langle\phi|$$
in weak-$\star$ sense.
\end{enumerate}
\end{lemma}

\begin{proof}
We shall show that
$$\lim_{N\to\infty}\langle\Psi,\left(\widehat{n}^\phi\right)^l\Psi\rangle=0
\Rightarrow
\lim_{N\to\infty}\left\langle\Psi,\left(\widehat{n}^\phi\right)^j\Psi\right\rangle=0
$$ for any $j,l>0$, which is equivalent to (a).

Let $\lim_{N\to\infty}\langle\Psi,(\widehat{n}^\phi)^j\Psi\rangle=0$
for some $j>0$. It follows, that there exists a function $\delta(N)$
with $\lim_{N\to\infty}\delta(N)=0$ such that
$$\sum_{k=0}^N
\left(\frac{k}{N}\right)^j\|P_k\Psi\|<\delta(N)\;.$$ Let $k(N)$ be
the smallest integer such that
$\left(\frac{k(N)}{N}\right)^j<\sqrt{\delta(N)}$. It follows that
$\left(\frac{k(N)+1}{N}\right)^j\geq \sqrt{\delta(N)}$ and thus
$\sum_{k(N)+1}^N \|P_k\Psi\|<\sqrt{\delta(N)}$. Hence
\beas\sum_{k=0}^N \left(\frac{k}{N}\right)^l\|P_k\Psi\|&=&
\sum_{k=0}^{k(N)}
\left(\frac{k}{N}\right)^l\|P_k\Psi\|+\sum_{k(N)+1}^N \|P_k\Psi\|
\\&\leq&\left(\frac{k(N)}{N}\right)^l+\sqrt{\delta(N)}
\leq \left(\sqrt{\delta(N)}\right)^{l/j}+\sqrt{\delta(N)}\;.\eeas
Thus
$\lim_{N\to\infty}\left\langle\Psi,\left(\widehat{n}^\phi\right)^j\Psi\right\rangle=0$
and (a) follows.

With (a) we can choose without loss of generality $j=2$ to prove
(b). So let
$$\lim_{N\to\infty}\left\langle\Psi,(\widehat{n}^\phi)^2\Psi\right\rangle=0\;.$$
With Lemma \ref{kombinatorik} (c) we have using symmetry of $\Psi_t$
that $\lim_{N\to\infty} \|q_1^\phi\Psi\|=0$ and $\lim_{N\to\infty}
\|p_1^\phi\Psi\|=1$. Note, that \beas \mu(\Psi)&=&\int
\Psi(\cdot,x_2,\ldots,x_N)\Psi^*(\cdot,x_2,\ldots,x_N)d^{3N-3}x
\\&=&\int
p_1^\phi\Psi(\cdot,x_2,\ldots,x_N)p_1^\phi\Psi^*(\cdot,x_2,\ldots,x_N)d^{3N-3}x
\\&&+\int
q_1^\phi\Psi(\cdot,x_2,\ldots,x_N)p_1^\phi\Psi^*(\cdot,x_2,\ldots,x_N)d^{3N-3}x
\\&&+\int
p_1^\phi\Psi(\cdot,x_2,\ldots,x_N)q_1^\phi\Psi^*(\cdot,x_2,\ldots,x_N)d^{3N-3}x
\\&&+\int
q_1^\phi\Psi(\cdot,x_2,\ldots,x_N)q_1^\phi\Psi^*(\cdot,x_2,\ldots,x_N)d^{3N-3}x
\eeas The first summand equals $\|p_1^\phi\Psi\|^2\;
|\phi\rangle\langle\phi|$, the other summands have operator norm
$\|q_1^\phi\Psi\|\;\|p_1^\phi\Psi\|$ and $\|q_1^\phi\Psi\|^2$
respectively and the Lemma follows.

\end{proof}

\begin{remark}
Similarly one can proof that
$\lim_{N\to\infty}\left\langle\Psi,(\widehat{n}^\phi)^\gamma\Psi\right\rangle=0$
for $\gamma\in\mathbb{R}^+$ implies convergence of the reduced
$k$-particle density matrix for any fixed $k<\infty$.
\end{remark}

\section{Derivation of the Hartree equation}

Let us now consider the different cases for $\beta$. To motivate the
technique we shall use below, we first take a short look at
$\beta=0$. In this case we have that the mean field is of the form
$v\star |\phi_t|^2$ and (\ref{meanfield}) becomes the Hartree
equation.

Let $\phi_t$ be a solution of the Hartree equation, let $T<\infty$
be such that $\|\phi_t\|_\infty<\infty$ for all $t<T$.

Defining
\be\label{defalpha1}\alpha_t:=\|\widehat{n}^{\phi_t}\Psi_t\|=\langle\Psi_t,(\widehat{n}^{\phi_t})^2\Psi_t\rangle\ee
and assuming that $\alpha_0\to 0$ as $N\to\infty$ we wish to show
that $\alpha_t\to0$ uniform in $t<T$.

Note, that $\alpha_t$ is $1/N$ times the expectation of particles
which are not in the state $\phi_t$, i.e.
$1-\alpha_t=\langle\Psi_t,\left(1-(\widehat{n}^{\phi_t})^2\right)\Psi_t\rangle$
is $1/N$ times the expectation of particles which are in the state
$\phi_t$.

By (\ref{defalpha1}) \beas \alpha_t'=\frac{d}{dt}\alpha_t:=-i
\langle\Psi_t,[H-H^H,(\widehat{n}^{\phi_t})^2]\Psi_t\rangle\eeas
where
$$H^H:=\sum_{j=1}^N -\Delta_j+A_t(x_j)+(v\star |\phi_t|^2)(x_j)\;.$$ Using symmetry of $\Psi_t$ and Definition \ref{defpro} we have
\beas \alpha_t'&=&-iN^{-1}\sum_{j=1}^N \langle \Psi_t,[\sum_{k\neq
l}v^N_{\beta}(x_k-x_l)-\sum_{l=1}^Nv\star
|\phi_t|^2(x_l),q^{\phi_t}_j] \Psi_t\rangle
\\&=& -i\langle
\Psi_t,[\sum_{k\neq 1}v^N_{\beta}(x_k-x_1)-v\star
|\phi_t|^2(x_1),q^{\phi_t}_1] \Psi_t\rangle
%
%
%
%
%
%
%
\\&=& -i\langle \Psi_t,\big((N-1)v^N_{\beta}(x_2-x_1)-v\star
|\phi_t|^2(x_1)\big)q^{\phi_t}_1\Psi_t\rangle\\&&+i\langle
\Psi_t,q^{\phi_t}_1\big(v^N_{\beta}(x_2-x_1)-v\star
|\phi_t|^2(x_1)\big)\Psi_t\rangle
\\&=& -i\big(\langle
\Psi_t,q^{\phi_t}_1\big((N-1)v^N_{\beta}(x_2-x_1)-v\star
|\phi_t|^2(x_1)\big)q^{\phi_t}_1\Psi_t\rangle\\&&-i\langle
\Psi_t,p^{\phi_t}_1\big((N-1)v^N_{\beta}(x_2-x_1)-v\star
|\phi_t|^2(x_1)\big)q^{\phi_t}_1\Psi_t\rangle\\&&+i\langle
\Psi_t,q^{\phi_t}_1\big((N-1)v^N_{\beta}(x_2-x_1)-v\star
|\phi_t|^2(x_1)\big)q^{\phi_t}_1\Psi_t\rangle\\&&+i\langle
\Psi_t,q^{\phi_t}_1\big((N-1)v^N_{\beta}(x_2-x_1)-v\star
|\phi_t|^2(x_1)\big)p^{\phi_t}_1\Psi_t\rangle\big)
 \eeas
Using selfadjointness of the multiplication operators the first and
third summand cancel out and we get \beas |\alpha_t'|&\leq&
2|\langle \Psi_t,p^{\phi_t}_1\big((N-1)v^N_{\beta}(x_2-x_1)-v\star
|\phi_t|^2(x_1)\big)q^{\phi_t}_1\Psi_t\rangle| \eeas Using $\langle
\Psi, p^{\phi_t}_2 v^N_{\beta}(x_1-x_2)
p^{\phi_t}_2\Psi\rangle=\langle\Psi,(v\star
|\phi_t|^2)(x_1)p^{\phi_t}_2\Psi\rangle $ and Lemma
\ref{kombinatorik} (d)
 \beas
|\alpha_t'|&\leq& 2|\langle
\Psi_t,p^{\phi_t}_1p^{\phi_t}_2\big((N-1)v^N_{\beta}(x_2-x_1)-v\star
|\phi_t|^2(x_1)\big)q^{\phi_t}_1q^{\phi_t}_2\Psi_t\rangle|
\\&&+2|\langle \Psi_t,p^{\phi_t}_1q^{\phi_t}_2\big((N-1)v^N_{\beta}(x_2-x_1)-v\star
|\phi_t|^2(x_1)\big)q^{\phi_t}_1p^{\phi_t}_2\Psi_t\rangle|
\\&&+2|\langle \Psi_t,p^{\phi_t}_1q^{\phi_t}_2\big((N-1)v^N_{\beta}(x_2-x_1)-v\star
|\phi_t|^2(x_1)\big)q^{\phi_t}_1q^{\phi_t}_2\Psi_t\rangle|
\\&\leq& 2|\langle
\Psi_t,\widehat{n}_2^{\phi_t}p^{\phi_t}_1p^{\phi_t}_2\big((N-1)v^N_{\beta}(x_2-x_1)-v\star
|\phi_t|^2(x_1)\big)(\widehat{n}^{\phi_t})^{-1}q^{\phi_t}_1q^{\phi_t}_2\Psi_t\rangle|
\\&&+2\left(\left\|p^{\phi_t}_1q^{\phi_t}_2\Psi_t\right\|^2+\left\|p^{\phi_t}_1q^{\phi_t}_2\Psi_t\right\|\;\left\|q^{\phi_t}_1q^{\phi_t}_2\Psi_t\right\|\right)\left(\left\|(N-1)v^N_{\beta}\right\|_\infty+\left\|v\star
|\phi_t|^2\right\|_\infty\right)\;.
 \eeas

Remember that in the case $\beta=0$ the scaling is such that
$v^N_{\beta}=N^{-1}v$, thus $\|v^N_{\beta}\|_1=N^{-1}\|v\|_1$ and
$\|v^N_{\beta}\|_\infty=N^{-1}\|v\|_\infty$. Note also that
$\frac{\sqrt{j+2}}{\sqrt{N}}<\frac{\sqrt{j}}{\sqrt{N}}+\frac{2}{\sqrt{N}}$
and thus $\widehat{N}_2^{\phi_t}\leq
\widehat{n}^{\phi_t}+\frac{2}{\sqrt{n}}$. It follows that
 \beas
|\alpha_t'|&\leq&
C\left\|\left(\widehat{n}^\phi_t+\frac{2}{\sqrt{N}}\right)\Psi_t\right\|\;\|(\widehat{n}^{\phi_t})^{-1}q^{\phi_t}_1q^{\phi_t}_2\Psi_t\|
\\&&+2C\left(\left\|p^{\phi_t}_1q^{\phi_t}_2\Psi_t\right\|^2+\left\|p^{\phi_t}_1q^{\phi_t}_2\Psi_t\right\|\;
\left\|q^{\phi_t}_1q^{\phi_t}_2\Psi_t\right\|\right)\;.
 \eeas
Using Lemma \ref{kombinatorik} (c) it follows in view of
(\ref{defalpha1}) that one can find a $C<\infty$ such that
$$|\alpha_t'|\leq C\alpha+C N^{-1/2}\;,$$ thus by Gronwalls Lemma $\alpha_t\to 0$ for
$N\to\infty$ uniform in $t<T$ (under the assumptions above, in
particular $\alpha_0\to0$ for $N\to\infty$).

\section{Derivation of the Gross-Pitaevskii equation}

Let us now consider the case $0<\beta\leq 1$. Then (\ref{meanfield})
becomes the Gross Pitaevskii equation \be\label{GP} i\frac{d}{dt}
\phi^{GP}_t=\left(-\Delta +A_t\right) \phi^{GP}_t+
2a|\phi^{GP}_t|^2\phi^{GP}_t:=h^{GP}\phi^{GP}_t\;.\ee The respective
Gross Pitaevskii energy is given by \bea\label{energyfunct}
E^{GP}_t&=&E_{t}^{kin}+E_{t}^{pot}:=\langle\nabla\phi^{GP}_t,\nabla\phi^{GP}_t\rangle+\langle\phi^{GP}_t,(A_t+a|\phi^{GP}_t|^2)\phi^{GP}_t\rangle
\nonumber\\&=&\langle\phi^{GP}_t,
(h^{GP}-a|\phi^{GP}_t|^2)\phi^{GP}_t\rangle\;.
\eea To control
$\langle\Psi_t,\widehat{n}^{\phi^{GP}_t}\Psi_t\rangle$, the
solutions $\phi^{GP}_t$ of the Gross Pitaevskii equation we shall
consider have to satisfy some additional conditions. If we have in
addition sufficiently strong decay conditions on $\phi^{GP}_t$ in
$t$ we can even get control on the respective $\alpha_t$ uniform in
$t<\infty$. Therefore we shall define next the sets $\mathcal{G}$
and $\mathcal{G}_{dec}$ of solutions of (\ref{GP}) which satisfy
these conditions

\begin{definition}\label{defGP}
$$\mathcal{G}:=\{\phi^{GP}_t:i\frac{d}{dt} \phi^{GP}_t=h^{GP}\phi^{GP}_t;\;\|\phi^{GP}_t\|_\infty+\|\nabla\phi^{GP}_t\|_\infty+\|\Delta\phi^{GP}_t\|_\infty<\infty\;\forall\;t\geq0\}$$
and
$$\mathcal{G}_{dec}:=\{\phi^{GP}_t\in\mathcal{G}: \int_0^\infty\|\phi^{GP}_t\|_\infty+\|\nabla\phi^{GP}_t\|_\infty dt<\infty\}$$
\end{definition}

Furthermore we shall --- depending on $\beta$ --- need some
conditions on the interaction $v^N_{\beta}$. These conditions shall
include the potentials we used in the introduction, i.e. potentials
which scale like $v^N_{\beta}(x)= N^{-1+3\beta} v(N^\beta x)$ for a
compactly supported, spherically symmetric, positive potential $v\in
L^1\cap L^\infty$.

\begin{definition}\label{defpot}
For any $0<\beta\leq1$ let
$$\mathcal{W}_\beta:=\{v^N_{\beta}\text{ pos. and spher. symm. } v^N_{\beta}(x)=0\;\forall\;x>R N^{-\beta}\text{ for some }R<\infty \}\;.$$
For any $0<\beta<1$ let \beas\mathcal{V}_{\beta}:=\{v^N_{\beta}\in
\mathcal{W}_\beta:
&\lim_{N\to\infty}&N^{1-3\beta}\|v^N_{\beta}\|_\infty<\infty;
 \\&\lim_{N\to\infty}&N^{1+\delta}(\|v^N_{\beta}\|_1-a/N)<\infty\text{ for some
}\delta>0\} \eeas and let
$$\mathcal{V}_{1}:=\{v^N_{\beta}\in \mathcal{W}_1: \lim_{N\to\infty}N^{1+\delta}(scat(v^N_{\beta})-a/N)<\infty \text{ for some
}\delta>0\}\;,$$ where $scat(v)$ is the scattering length of the
potential $v$.
\end{definition}

With these definitions we arrive at the main Theorem:

\begin{theorem}\label{theorem}
Let $0<\beta\leq1$, let $v^N_{\beta}\in\mathcal{V}_{\beta}$ and let
$\phi^{GP}_t\in\mathcal{G}$. Let $T<\infty$ ($T\leq \infty$ if
$\phi^{GP}_t\in\mathcal{G}_{dec}$). Let $A_t$ be such that $\int_0^T
\|A_t'\|dt<\infty$. Let $\Psi_0$ be symmetric with $\|\Psi_0\|=1$,
\be\label{cond1}\lim_{N\to\infty}
N^{\delta}\left\langle\Psi_0,\left(\widehat{n}^{\phi^{GP}_t}\right)^2\Psi_0\right\rangle=0\ee
and \be\label{cond2}\lim_{N\to\infty}
N^{\delta}(N^{-1}\langle\Psi_0,H\Psi_0\rangle-E^{GP}_t)=0\ee for
some $\delta>0$.
 Then
\be\label{theoremeq}\lim_{N\to\infty}\left\langle\Psi_t,\left(\widehat{n}^{\phi^{GP}_t}\right)^2\Psi_t\right\rangle=0\ee
uniform in $0<t< T$.

\end{theorem}

\begin{remark}
\begin{enumerate}
\item Lemma \ref{kondensat}
implies  convergence of the reduced one-particle density matrix.

\item For $\beta=1$ the conditions on $v^N_{\beta}$ include the hard sphere case (and
potentials which scale like $v^N_{\beta}=N^\gamma v(Nx)$ for any
$\gamma>2$) with compactly supported $v$ with support-radius $a$:
Such potentials satisfy all conditions one needs to be in
$\mathcal{W}_1$ and the respective scattering length equals $a/N$
(converges against $a/N$) as $N\to\infty$.

\item It has been proven for a large class of external potentials  that the $N$-particle ground state
wave function $\Psi$ satisfies the conditions (\ref{cond1}) and
(\ref{cond2}) \cite{ls,lssy,lsy}. So the Theorem fits well for
describing the physics of a trapped, cooled Bose gas when the trap
is removed.

\item Condition (\ref{cond2}) can be understood as smoothness condition on
$\Psi_0$. For the case $0<\beta<1$ this is clear on a heuristic
level: If all particles of $\Psi_0$ are more or less equal to
$\phi^{GP}_t$ and if $\Psi_0$ is smooth enough, then the energy of
$\Psi_0$ is of course close to $NE^{GP}_t$.

For $\beta=1$ note, that $L^2$ density arguments can be used, i.e.
if (\ref{theoremeq}) holds for some $\Psi_0$, then it also holds for
a sequence $\Psi_0^N$ which converges in $L^2$ against $\Psi_0$.
Thus we can equip $\Psi_0$ with a microscopic structure which does
not change the $L^2$ norm significantly in such a way, that the
energy gets close to $NE^{GP}_t$.

With the technique we shall present in this paper this can be done
rigorously.
\end{enumerate}

\end{remark}

\subsection{Proof of the Theorem}

\begin{notation}
In the following all projectors shall be with respect to
$\phi^{GP}_t$. We shall omit the upper index $\phi^{GP}_t$ on $p_j$,
$q_j$, $P_j$, $P_{j,k}$ and $\widehat{\cdot}$.
\end{notation}

Note that due to Lemma \ref{kondensat} (a) we have some flexibility
in choosing which term we wish to control: To prove the Theorem we
can choose to control
$\langle\Psi_t,(\widehat{n})^{\gamma}\Psi_t\rangle$ for arbitrary
$\gamma>0$ . We shall use $\gamma=1$ since we shall  estimate the
 kinetic energy (see Lemma \ref{ekin} below) in terms of
$\langle\Psi_t, \widehat{n}\;\Psi_t\rangle$.

\begin{definition}
Using the notation
$$h_{j,k}:=(N-1)v^N_{\beta}(x_j-x_k)-\frac{a}{2}|\phi_t^{GP}|^2(x_j)-\frac{a}{2}|\phi_t^{GP}|^2(x_k)$$
we define the functional $\alpha:\LZN\to\mathbb{R}^+$ by
$$\alpha(\Psi):=\langle\Psi,\widehat{n}\Psi\rangle=\|(\widehat{n})^{1/2}\Psi\|$$
and the functionals $\alpha_{1,2}':\LZN\to\mathbb{R}^+$ by
\bea\label{fnochda} \alpha_{1}'(\Psi)%
&=&N\Im\left(\langle\Psi ,h_{1,2} (\widehat{n}-\widehat{n}_2)p_1p_2
\Psi\rangle\right) \\
\label{fnochda2}
\alpha_{2}'(\Psi)%
&=&N\Im\left(\langle\Psi ,h_{1,2} (\widehat{n}-\widehat{n}_1)p_1q_2
\Psi\rangle\right) \;.\eea

\end{definition}

\begin{lemma}\label{ableitung}

For any solution of the Schr\"odinger equation $\Psi_t$ we have
$$\frac{d}{dt} \alpha(\Psi_t)=2\alpha_{1}'(\Psi_t)+4\alpha_{2}'(\Psi_t)\;.$$

\end{lemma}

\begin{proof}
We have for $0<\beta\leq 1$ for the time derivative
\beas\frac{d}{dt}\alpha(\Psi)&=&\frac{d}{dt}\langle\Psi
,\widehat{n}\;\Psi\rangle
\\&=&-i\langle H^\beta\Psi
,\widehat{n}\;\Psi\rangle+i\langle\Psi
,\widehat{n}\;H^\beta\Psi\rangle+i\langle \Psi
,[H^{GP}_t,\widehat{n}\;]\Psi\rangle
\\&=&-i\langle\Psi
,[H^\beta-H^{GP}_t,\widehat{n}\;]\Psi\rangle\;.
 \eeas

Using symmetry of $\Psi$ it follows that
\bea\label{wiehier} \frac{d}{dt}\alpha(\Psi)%
&=&-i(N-1)^{-1}\sum_{j\neq k}\langle\Psi
,[h_{j,k},\widehat{n}\;]\Psi\rangle
\\\nonumber&=&-iN\langle\Psi
,h_{1,2}\widehat{n}\;\Psi\rangle-\langle\Psi
,\widehat{n}\;h_{1,2}\Psi\rangle
=2N\Im\left(\langle\Psi ,h_{1,2}\widehat{n}\;\Psi\rangle\right)\;.
\eea Note that we can write for any
$m:\{1,\ldots,N\}\to\mathbb{R}^+$ \bea\label{nersetzen}
\widehat{m}&=&\sum_{k=0}^N m(k)P_k
\\\nonumber &=&\sum_{k=0}^{N-2}
\big(m(k)p_1  p_2  P_{N-2,k}+m(k)p_1 q_2 P_{N-2,k-1}\\\nonumber
&&+m(k)q_1  p_2 P_{N-2,k-1}+m(k)(1-p_1 q_2-q_1  p_2-p_1 p_2)
P_{N-2,k-2}\big)
\\\nonumber &=&\sum_{k=0}^{N}
\big(m(k)p_1  p_2  P_{N-2,k}+m(k)p_1 q_2 P_{N-2,k-1}\\\nonumber
&&+m(k)q_1  p_2 P_{N-2,k-1}+m(k)P_{N-2,k-2}\big)
\\\nonumber&&-\sum_{k=0}^{N}m(k+1)p_1  q_2P_{N-2,k-1}-m(k+1)q_1  p_2P_{N-2,k-1}\\\nonumber&&-m(k+2)p_1 p_2P_{N-2,k})
\\\nonumber &=&(\widehat{m}-\widehat{m}_2)p_1p_2+(\widehat{m}-\widehat{m}_1)p_1q_2+(\widehat{m}-\widehat{m}_1)q_1p_2+\sum_{k=0}^Nm(k)P_{N-2,k-2}
 \eea
Using  symmetry of $\Psi$ and selfadjointness of
$h_{1,2}P_{N-2,k-2}$ it follows that
$$ \frac{d}{dt}\alpha(\Psi)%
=\Im\left(\langle\Psi ,h_{1,2}
\left(N(\widehat{n}-\widehat{n}_2)p_1p_2+2(\widehat{n}-\widehat{n}_1)p_1q_2\right)
\Psi\rangle\right) \;.$$
\end{proof}

\section{The Gross Pitaevskii equation for $0<\beta<1/3$}

In this section we shall control $\alpha_{1,\Psi}$ and
$\alpha_{2,\Psi}$ under additional conditions on $\beta$, namely
$\beta<1/3$ for $\alpha_{1,\Psi}$ and $\beta<1$ for
$\alpha_{2,\Psi}$.

\begin{lemma}\label{alphaabl}
We have under the conditions of Theorem \ref{theorem} that there
exists a $C<\infty$ and a $\xi>0$ such that for any $\Psi\in\LZN$
with $\nabla_1\Psi\in\LZN$ that
\begin{enumerate}
\item for $0<\beta<1/3$ $$
|\alpha_{1}'(\Psi)|\leq
C(\|\phi_t^{GP}\|_\infty+\|\nabla\phi_t^{GP}\|_\infty)(\alpha(\Psi)+N^{-\xi})$$
\item for $0<\beta<1$$$
|\alpha_{2}'(\Psi)|\leq
C(\|\phi_t^{GP}\|_\infty+\|\nabla\phi_t^{GP}\|_\infty)(\alpha(\Psi)+\|\nabla_1q_1\Psi\|+N^{-\xi})$$

\end{enumerate}
\end{lemma}
\begin{proof}
Using (\ref{fnochda}) and $1=p_1 p_2+p_1 q_2+q_1 p_2+q_1 q_2$
\beas\alpha_{1}'(\Psi)&=&
 \Im\left(\langle\Psi ,p_1p_2h_{1,2}
N(\widehat{n}-\widehat{n}_2)p_1p_2 \Psi\rangle\right)
+\Im\left(\langle\Psi ,p_1q_2h_{1,2}
N(\widehat{n}-\widehat{n}_2)p_1p_2 \Psi\rangle\right)
\\&&+\Im\left(\langle\Psi ,q_1p_2h_{1,2}
N(\widehat{n}-\widehat{n}_2)p_1p_2 \Psi\rangle\right)
+\Im\left(\langle\Psi ,q_1q_2h_{1,2}
N(\widehat{n}-\widehat{n}_2)p_1p_2 \Psi\rangle\right)
\\\alpha_{2}'(\Psi)&=&
 \Im\left(\langle\Psi ,p_1p_2h_{1,2}
N(\widehat{n}-\widehat{n}_1)p_1q_2 \Psi\rangle\right)
+\Im\left(\langle\Psi ,p_1q_2h_{1,2}
N(\widehat{n}-\widehat{n}_1)p_1q_2 \Psi\rangle\right)
\\&&+\Im\left(\langle\Psi ,q_1p_2h_{1,2}
N(\widehat{n}-\widehat{n}_1)p_1q_2 \Psi\rangle\right)
+\Im\left(\langle\Psi ,q_1q_2h_{1,2}
N(\widehat{n}-\widehat{n}_1)p_1q_2 \Psi\rangle\right)\;. \eeas Using
that
$\Im(\langle\Psi,A\Psi\rangle)=-\Im(\langle\Psi,A^t\Psi\rangle)$ for
any operator $A$ and that $\Psi$ is symmetric (note that $p_1
q_2h_{1,2}q_1  p_2$ is invariant under adjunction plus exchange of
the variable $x_1$ and $x_2$) and Lemma \ref{kombinatorik} (dc) we
get \beas\alpha_{1}'(\Psi)&=
2\Im\left(\langle\Psi ,p_1q_2h_{1,2}
N(\widehat{n}-\widehat{n}_2)p_1p_2 \Psi\rangle\right)
+\Im\left(\langle\Psi ,q_1q_2h_{1,2}
N(\widehat{n}-\widehat{n}_2)p_1p_2 \Psi\rangle\right)
\\\alpha_{2}'(\Psi)&=
 \Im\left(\langle\Psi ,N(\widehat{n}_1-\widehat{n}_2)p_1p_2h_{1,2}
p_1q_2 \Psi\rangle\right)
+\Im\left(\langle\Psi ,q_1q_2h_{1,2}
N(\widehat{n}-\widehat{n}_1)p_1q_2 \Psi\rangle\right) \eeas

Note that
\beas\sqrt{k/N}-\sqrt{(k-2)/N}&=&\left(k/N-(k-2)/N\right)/\left(\sqrt{k/N}+\sqrt{(k-2)/N}\right)\\&\leq&
(2/N)/(\sqrt{k/N})=2(Nk)^{-1/2}\;, \eeas so we have that
$0\leq(\widehat{n}-\widehat{n}_{1}),(\widehat{n}_1-\widehat{n}_{2})\leq(\widehat{n}-\widehat{n}_{2})\leq2(N\widehat{n})^{-1}$
 and Lemma \ref{alphaabl} follows from
\begin{lemma}\label{L2absch}
Let $m:\{1,\ldots,N\}\to\mathbb{R}^+$ with $m\leq n^{-1}$,
$0<\beta<1$. Then we have under the conditions of the Theorem that
there exists a $C<\infty$ and a $\xi>0$ such that
\begin{enumerate}

\item
for any $0<\beta<1$ \beas
&&|\langle\Psi_t ,p_1p_2 h_{1,2} \widehat{m}q_1p_2
 \Psi_t\rangle|\leq C(\|\phi_t^{GP}\|_\infty+\|\nabla\phi_t^{GP}\|_\infty)N^{-\xi}
%
\\ &&|\langle\Psi_t ,p_1 q_2 \widehat{m}^{1/2} h_{1,2} \widehat{m}^{1/2}q_1 q_2
\Psi_t\rangle|\\&&\hspace{2cm}\leq
C(\|\phi_t^{GP}\|_\infty+\|\nabla\phi_t^{GP}\|_\infty)
(\alpha(\Psi)+N^{-\gamma}+\|\nabla_1q_1\Psi\|^2)
\eeas

\item for any $0<\beta<1/3$
 \beas
|\langle\Psi_t ,p_1  p_2\widehat{m}^{1/2}  h_{1,2}
\widehat{m}^{1/2}q_1 q_2 \Psi_t\rangle| \leq
C(\|\phi_t^{GP}\|_\infty+\|\nabla\phi_t^{GP}\|_\infty)
(\alpha(\Psi_t)+N^{-\xi}) \eeas

\end{enumerate}
\end{lemma}

The proof of Lemma \ref{L2absch} shall be given in the Appendix for
later reference in a more general form. \end{proof}

\subsection{Control of the  kinetic energy for $\beta<1$}


To finish the control of $\alpha(\Psi_t)$ we shall provide a
sufficient estimate on the kinetic energy of $\Psi_t$, in particular
$\|\nabla_1q_1\Psi_t\|$. This estimate shall be given in terms of
$\alpha(\Psi_t)$, thus finally our estimate on $\alpha'(\Psi_t)$
shall depend on $\alpha_{\Psi}$ making $\alpha(\Psi_t)$ controllable
by a Gronwall argument. For that we need

\begin{lemma}\label{L2absch2}
Let $m:\{1,\ldots,N\}\to\mathbb{R}^+$ with $m\leq n^{-1}$,
$0<\beta<1$. Then we have under the conditions of the Theorem that
there exists a $C<\infty$ and a $\xi>0$ such that for any
$0<\beta<1$

\beas |\langle\Psi ,p_1p_2
\left((N-1)v^N_{\beta}(x_1,x_2)-a|\phi^{GP}_t|^2(x_1)\right)p_1p_2
 \Psi\rangle|\leq C(\|\phi_t^{GP}\|_\infty+\|\nabla\phi_t^{GP}\|_\infty)N^{-\xi}\eeas

and \beas |\langle\Psi ,p_1 p_2 v^N_{\beta}(x_1,x_2) q_1 q_2
\Psi\rangle|\leq
CN^{-1}(\|\phi_t^{GP}\|_\infty+\|\nabla\phi_t^{GP}\|_\infty)(\alpha(\Psi)+N^{-\xi})\;.
\eeas \end{lemma}

The proof of which shall be given together with the proof of Lemma
\ref{L2absch} in the Appendix.

\begin{lemma}\label{ekin}
Let $0<\beta<1$. Then we have under the conditions of Theorem
\ref{theorem} that there exists a $\xi>0$ such that uniform in $0<t<
T$
$$\|\nabla_1q_1\Psi_t\|^2\leq C\left(\sup_{0\leq s\leq
t}\{\alpha(\Psi_s)\}+N^{-\xi}\right)$$

\end{lemma}

\begin{proof}
 Using
symmetry of $\Psi_t$
$$N^{-1}\langle\Psi_t,H\Psi_t\rangle=-\|\nabla_1\Psi_t\|^2+(N-1)\langle\Psi_t,v^N_\beta(x_1-x_2)\Psi_t\rangle+\langle\Psi_t,A_t(x_1)\Psi_t\rangle\;,$$
Thus \beas
&&\|\nabla\phi^{GP}_t\|^2-\|\nabla_1\Psi_t\|^2=N^{-1}\langle\Psi_t,H\Psi_t\rangle-E^{GP}_t-\langle
\Psi_t,\left(A_t(x_1)+
a|\phi_t^{GP}|^2(x_1)\right)\Psi_t\rangle\\&&+\langle \phi_t^{GP},
\left(A_t+a|\phi_t^{GP}|^2\right)\phi_t^{GP}\rangle
+\left\langle\Psi_t,\left((N-1)v^N_{\beta}(x_1-x_2)-
a|\phi^{GP}_t|^2(x_1)\right)\Psi_t\right\rangle\;.\eeas

Using symmetry of $\Psi_t$ \beas&&
\frac{d}{dt}\left(N^{-1}\langle\Psi_t,H\Psi_t\rangle-E^{GP}_t\right)
=\langle\Psi_t,A_t'(x_1)\Psi_t\rangle-\langle\phi^{GP}_t,a\left(\frac{d}{dt}|\phi^{GP}_t|^2\right)\phi^{GP}_t\rangle
\\&&-\langle\phi^{GP}_t,A_t'\phi^{GP}_t\rangle-\langle\phi^{GP}_t,[(h^{GP}-a|\phi_t^{GP}|^2),h^{GP}]\phi^{GP}_t\rangle
\\&=&\langle\Psi_t,A_t'(x_1)\Psi_t\rangle-\langle\phi^{GP}_t,A_t'\phi^{GP}_t\rangle
+\langle\phi^{GP}_t,[a|\phi_t^{GP}|^2,h^{GP}]\phi^{GP}_t\rangle\\&&-\langle\phi^{GP}_t,[a|\phi_t^{GP}|^2,h^{GP}]\phi^{GP}_t\rangle
\\&\leq&4\|A_t'\|_\infty\left(N^{-1/4}+\alpha(\Psi_t)\right)\;,
\eeas where we used Lemma \ref{kombinatorik} (e) in the last step.
It follows using condition (\ref{cond2}) that for $N$ sufficiently
small (i.e. such that
$N^\delta\left(N^{-1}\langle\Psi_0,H\Psi_0\rangle-E^{GP}_0\right)<1$)
\beas\left(N^{-1}\langle\Psi_t,H\Psi_t\rangle-E^{GP}_t\right)&<&N^{-\delta}+\int_{0}^t
\|A_s'\|_\infty\left(\left(\frac{2}{N}\right)^{1/4}+\alpha(\Psi_s)\right)ds
\\&\leq&C(N^{-\delta}+N^{-1/4}+\sup_{0\leq s\leq
t}\{\alpha(\Psi_s)\}) \eeas uniform in $t< T$. Note that due to
Lemma \ref{kombinatorik} (e)
 \beas&&\langle \Psi_t,\left(A_t(x_1)+
a|\phi_t^{GP}|^2(x_1)\right)\Psi_t\rangle-\langle \phi_t^{GP},
\left(A_t+a|\phi_t^{GP}|^2\right)\phi_t^{GP}\rangle
\\&\leq& 4(\|A_t\|_\infty+a\|\phi_t^{GP}\|^2_\infty)\left(N^{-1/4}+\alpha(\Psi_t)\right)\;,\eeas
thus \bea\label{kinensum}
\left|\|\nabla_1\Psi_t\|^2-\|\nabla\phi^{GP}_t\|^2\right|&\leq&
C(N^{-\delta}+N^{-1/4}+\sup_{0\leq s\leq t}\{\alpha(\Psi_s)\})
\\\nonumber&&+\left|\left\langle\Psi_t,\left((N-1)v^N_{\beta}(x_1-x_2)-
a|\phi^{GP}_t|^2(x_1)\right)\Psi_t\right\rangle\right|\;.\eea
 We get using  symmetry
of $\Psi_t$ and self adjointness of the multiplication operators for
the last summand in (\ref{kinensum}) \bea\label{positiverterm}
\nonumber&&\left\langle
\Psi_t,\left((N-1)v^N_{\beta}(x_1-x_2)-a|\phi^{GP}_t|^2(x_2)\right)\Psi_t\right\rangle
\\\nonumber&=&\left\langle p_1p_2\Psi_t,\left((N-1)v^N_{\beta}(x_1-x_2)-a|\phi^{GP}_t|^2(x_2)\right)p_1p_2\Psi_t\right\rangle
\\\nonumber&&+2\Re\left\langle p_1p_2\Psi_t,(N-1)v^N_{\beta}(x_1-x_2)-a|\phi^{GP}_t|^2(1-p_1p_2)\Psi_t\right\rangle
\\&&+(N-1)\left\langle (1-p_1p_2)\Psi_t,v^N_{\beta}(x_1-x_2)(1-p_1p_2)\Psi_t\right\rangle
\\\nonumber&&-a\left\langle (1-p_1p_2)\Psi_t,|\phi^{GP}_t|^2(1-p_1p_2)\Psi_t\right\rangle
\;. \eea Using symmetry of $\Psi_t$, the absolute value of the
second term is bounded by
 \beas &&4(N-1)\left|\left\langle
p_1p_2\Psi_t,(N-1)v^N_{\beta}(x_1-x_2)-a|\phi^{GP}_t|^2p_1q_2\Psi_t\right\rangle\right|
\\+&&2(N-1)\left|\left\langle
p_1p_2\Psi_t,v^N_{\beta}(x_1-x_2)q_1q_2\Psi_t\right\rangle\right|
\eeas Using Lemma \ref{L2absch} in its more general form as given in
the Appendix and using positivity of $v^N_{\beta}$ (implying
positivity of line (\ref{positiverterm})) we get that \beas
&&\left\langle
\Psi_t,\left(\sum_{j\neq1}v^N_{\beta}(x_j-x_1)-a|\phi^{GP}_t|^2\right)\Psi_t\right\rangle
\\&\geq&-C(\alpha(\Psi_t)+N^{-\xi})
-a\|\phi^{GP}_t\|_\infty^2\;\|(1-p_1p_2)\Psi_t\|^2
\;. \eeas Writing $1-p_1p_2=p_1q_2+q_1p_2+q_1q_2$ Lemma
\ref{kombinatorik} yields \be\label{kinenzit}\left\langle
\Psi_t,\left(\sum_{j\neq1}v^N_{\beta}(x_j-x_1)-a|\phi^{GP}_t|^2\right)\Psi_t\right\rangle\geq-C(\alpha(\Psi_t)+N^{-\xi}+N^{-\delta}+N^{\beta-1})\;,\ee
so with (\ref{kinensum}) \beas
\left|\|\nabla_1\Psi_t\|^2-\|\nabla\phi^{GP}_t\|^2\right|&\leq&
C(N^{-\delta}+N^{-1/4}+\sup_{0\leq s\leq
t}\{\alpha(\Psi_s)\}+N^{-\beta}+N^{\beta-1}) \;.\eeas Note also that
$\|\nabla \Psi_t\|^2=\|\nabla p_1\Psi_t\|^2+\|\nabla q_1\Psi_t\|^2$
and \beas \|\nabla p_1\Psi_t\|^2-\|\nabla
p_1\Psi_t\|^2&=&\|\nabla\phi^{GP}_t\|^2(\|
p_1\Psi_t\|^2-1)=\|\nabla\phi^{GP}_t\|^2\|q_1\Psi_t\|^2\\&\leq&
\|\nabla\phi^{GP}_t\|^2\|\widehat{n}\;\Psi_t\|^2\leq
\|\nabla\phi^{GP}_t\|^2\alpha(\Psi_t)\;.\eeas Choosing
$\xi\leq\min\{\delta,1/4,\beta,1-\beta\}$ Lemma \ref{ekin} follows.

\end{proof}


\subsection{Proof of Theorem \ref{theorem} for $\beta<1/3$}

Lemma \ref{ableitung} with Lemma \ref{alphaabl} and Lemma \ref{ekin}
gives \be\label{estalpha}|\alpha'(\Psi_t)|\leq
C(\|\phi_t^{GP}\|_\infty+\|\nabla\phi_t^{GP}\|_\infty)(\sup_{0\leq
s\leq t}\{\alpha(\Psi_s)\}+N^{-\xi})\;.\ee We shall use a Gronwall
argument to control $\alpha(\Psi_t)$: Consider the differential
equation \be\label{diffeq} \gamma_t'=
C(\|\phi_t^{GP}\|_\infty+\|\nabla\phi_t^{GP}\|_\infty)(\sup_{0\leq
s\leq t}\{\gamma_s\}+N^{-\xi})\;.\ee Since the right hand side of
(\ref{diffeq}) is positive, the solution $\gamma_t$ with
$\gamma_0=\alpha(\Psi_0)$ dominates $\alpha(\Psi_t)$. Moreover
 $\gamma_t$
increases monotonously, thus $\sup_{0\leq s\leq
t}\{\gamma_s\}=\gamma_t$ and
$$\gamma_t'=
C(\|\phi_t^{GP}\|_\infty+\|\nabla\phi_t^{GP}\|_\infty)(\gamma_t+N^{-\xi})\;.$$
It follows that
$$\ln(\gamma_t+N^{-\xi})=C\int_0^t(\|\phi_s^{GP}\|_\infty+\|\nabla\phi_s^{GP}\|_\infty)ds+K_N$$
where the integration constant $K_N$ is such that
$\gamma_0=\alpha(\Psi_0)$, i.e.
$$\gamma_t=C_N\exp\left(C\int_0^t\|\phi_s^{GP}\|_\infty+\|\nabla\phi_s^{GP}\|_\infty ds\right)-N^{-\xi}$$
where $C_N=e^{K_N}$. Note that Lemma \ref{kondensat} implies with
(\ref{cond1}) that $\lim_{N\to\infty}\alpha(\Psi_0)=0$, thus
$\lim_{N\to\infty}C_N=0$ and Theorem \ref{theorem} follows for
$0<\beta<1$. \begin{flushright} $ \Box $ \end{flushright}

\section{The Gross Pitaevskii equation for $1/3\leq \beta\leq 1$}

\subsection{Microscopic Structure}

\begin{definition}\label{microscopic}
Let $0<\beta_1<\beta_2<1$, $v^N_{\beta_2}\in\mathcal{V}_{\beta_2}$.
We define the potential $W_{\beta_1,\beta_2}$ via
$$W^N_{\beta_1,\beta_2}(x):=\left\{
  \begin{array}{ll}
    aN^{-1+3\beta_1}, & \hbox{ for $RN^{-\beta_2}<x< R^N_{\beta_1,\beta_2}$;} \\
    0, & \hbox{else.}
  \end{array}
\right.$$ Here $RN^{-\beta_2}$ is an upper bound on the radius of
the support of $v^N_{\beta_2}$ (see Definition \ref{defpot}) and
$R^N_{\beta_1,\beta_2}$ is the minimal value which ensures that the
scattering length of $v^N_{\beta_2}-W^N_{\beta_1,\beta_2}$ is zero.

The respective zero energy scattering state shall be denoted by
$f^N_{\beta_1,\beta_2}$, i.e.
$$\left(-\Delta+v^N_{\beta_2}-W^N_{\beta_1,\beta_2}\right)f^N_{\beta_1,\beta_2}=0\;,$$
we shall also need $$g^N_{\beta_1,\beta_2}=1-f^N_{\beta_1,\beta_2}$$
\end{definition}

\begin{lemma}\label{defAlemma}
For any $0<\beta_1<\beta_2\leq1$,
$v^N_{\beta_2}\in\mathcal{V}_{\beta_2}$

\begin{enumerate}
\item $$\|g^N_{\beta_1,\beta_2}\|\leq \sqrt{8\pi} a N^{-1-\beta_1/2}\;\;\;\;\;\|g^N_{\beta_1,\beta_2}\|_1\leq 16\pi a N^{-1-2\beta_1}$$

\item
$$W^N_{\beta_1,\beta_2}f^N_{\beta_1,\beta_2}\in\mathcal{V}_{\beta_1}$$

\item
The operator $h:=-\Delta+v^N_{\beta}-W^N_{\beta_1,\beta_2}$ is
positive.

\item
For any $\beta<\gamma<1$ let $B_{\gamma}:=\{x\in\mathbb{R}^3:|x|\leq
N^{-\gamma}\}$. Then for any $\Psi\in \mathcal{D}(H)$
$$\|\mathds{1}_{B_\gamma}\nabla\Psi\|+\langle\Psi,
(v^N_{\beta}-W^N_{\beta_1,\beta_2})\Psi\rangle\geq0$$

\end{enumerate}
\end{lemma}

\begin{proof}
Let $j^N_{\beta_2}$ be the zero energy scattering state of the
potential $\frac{1}{2}v^N_{\beta_2}$. Since $v^N_{\beta_2}$ is
positive and has compact support of radius $r_N$ it follows, that
$1>j^N_{\beta_2}(x)\geq 1-a/(Nx)$
 for
any $x\geq r_N$. Note, that the potential $W^N_{\beta_1,\beta_2}$ is
zero inside the Ball around zero of radius $RN^{-\beta_2}$, hence
$f^N_{\beta_1,\beta_2}$ is inside this Ball a multiple of
$j^N_{\beta_2}$.

 Let $R^N_{\beta_1,\beta_2}$ be such that
$K^N_{\beta_1,\beta_2}f^N_{\beta_1,\beta_2}(x)=j^N_{\beta_2}(x)$ for
any $x<N^{-\beta_1}$. By definition of the potential
$W^N_{\beta_1,\beta_2}$ we have that $\partial_x
f^N_{\beta_1,\beta_2}(x)\geq0$: $R^N_{\beta_1,\beta_2}$ was defined
to be the minimal value which ensures that the scattering length of
$v^N_{\beta}-W^N_N$ is zero, thus $\partial_r f^N_N(r)\geq 0$  for
$r<R^N_{\beta_1,\beta_2}$. It follows in particular that
$f^N_{\beta_1,\beta_2}\leq 1$. Furthermore we have, since
$W^N_{\beta_1,\beta_2}$ is positive, that
$K^N_{\beta_1,\beta_2}\partial_rf^N_{\beta_1,\beta_2}\leq
\partial_rj^N_{\beta_2}$ and $K^N_{\beta_1,\beta_2}f^N_{\beta_1,\beta_2}\leq j^N_{\beta_2}$.

Since $f^N_{\beta_1,\beta_2}(x)=1$ for $x>2N^{-\beta_1}$ and
$\lim_{x\to\infty}j^N_{\beta_2}(x)=1$ we get that
$K^N_{\beta_1,\beta_2}\leq 1$, thus $1>f^N_{\beta_1,\beta_2}\geq
j^N_{\beta_2}$. Since $j^N_{\beta_2}(x)\geq 1-a/(Nx)$ it follows
that \be\label{gbound}|g^N_{\beta_1,\beta_2}(x)|\leq a/(Nx)\;.\ee
Since $g^N_{\beta_1,\beta_2}(x)=0$ for $x>2N^{-\beta_1}$  it follows
that \beas \|g^N_{\beta_1,\beta_2}\|^2&\leq&
a^2N^{-2}\int_{0}^{2N^{-\beta_1}} |x|^{-2}d^3x=8N^{-\beta_1}\pi
a^2N^{-2}
\\\|g^N_{\beta_1,\beta_2}\|_1&\leq& aN^{-1}\int_{0}^{2N^{-\beta_1}}
|x|^{-1}d^3x=16N^{-2\beta_1}\pi aN^{-1}  \eeas which is (a).

Next we have to show that
$W^N_{\beta_1,\beta_2}f^N_{\beta_1,\beta_2}\in\mathcal{V}_{\beta_1}$.

 Since $K^N_{\beta_1,\beta_2}f^N_{\beta_1,\beta_2}>1-aN^{-1-\beta_1}$ on the support of $W^N_{\beta_1,\beta_2}$
\bea\label{l1wn}\|W^N_{\beta_1,\beta_2}K^N_{\beta_1,\beta_2}f^N_{\beta_1,\beta_2}\|_1&<&\|W^N_{\beta_1,\beta_2}f^N_{\beta_1,\beta_2}\|_1<\|W^N_{\beta_1,\beta_2}\|_1
\nonumber\\&<&(1-aN^{-1+\beta_1}
)^{-1}\|W^N_{\beta_1,\beta_2}K^N_{\beta_1,\beta_2}f^N_{\beta_1,\beta_2}\|_1\nonumber\\&<&(1-aN^{-1+\beta_1}
)^{-1}\|W^N_{\beta_1,\beta_2}f^N_{\beta_1,\beta_2}\|_1\;.\eea

Note also that
$\|W^N_{\beta_1,\beta_2}K^N_{\beta_1,\beta_2}f^N_{\beta_1,\beta_2}\|_1=a/N$:
Read
$\rho^N_{\beta_1,\beta_2}:=W^N_{\beta_1,\beta_2}K^N_{\beta_1,\beta_2}f^N_{\beta_1,\beta_2}$
as a classical charge distribution which must compensate the charge
$a/N$ (recall that
$K^N_{\beta_1,\beta_2}f^N_{\beta_1,\beta_2}(x)=1-a/(Nx)$ for
$r_N<x<N^{-\beta}$) to get that the potential
$\varphi_N\widehat{=}K^N_{\beta_1,\beta_2}f^N_{\beta_1,\beta_2}$ is
zero outside the support of $W^N_{\beta_1,\beta_2}$. With
(\ref{l1wn}) it follows that
\be\label{ersterlimes}\lim_{N\to\infty}N^{1-\beta_1}(\|W^N_{\beta_1,\beta_2}f^N_{\beta_1,\beta_2}\|_1-a/N)<\infty\ee
and again using (\ref{l1wn})
$$\lim_{N\to\infty}N^{1-\beta_1}(\|W^N_{\beta_1,\beta_2}\|_1-a/N)<\infty\;.$$
It follows that the support of $W^N_{\beta_1,\beta_2}$ is of order
$N^{3\beta_1}$. Since $W^N_{\beta_1,\beta_2}f^N_{\beta_1,\beta_2}$
is spherically symmetric, positive and equal to zero for
$x>R^N_{\beta_1,\beta_2}$ it follows that
$W^N_{\beta_1,\beta_2}f^N_{\beta_1,\beta_2}\in
\mathcal{W}_{\beta_1}$.

With (\ref{ersterlimes}) and using that $W^N_{\beta_1,\beta_2}$ is
defined such that
$\|W^N_{\beta_1,\beta_2}f^N_{\beta_1,\beta_2}\|_\infty\leq\|W^N_{\beta_1,\beta_2}\|_\infty=aN^{-1+3\beta_1}$
it follows that
$W^N_{\beta_1,\beta_2}f^N_{\beta_1,\beta_2}\in\mathcal{V}_{\beta_1}$.

We show (c) by contradiction. Assume that $h$ is not positive, thus
it has a ground state $\chi$. Since $f^N_{\beta_1,\beta_2}$ is by
construction a positive function and so is the ground state $\chi$
it follows that $\int f^N_{\beta_1,\beta_2}(x)\chi^{*}(x)d^3x>0$.
But $f^N_{\beta_1,\beta_2}$ is the generalized eigenfunction of $h$
with energy $0$, so $\int f^N_{\beta_1,\beta_2}(x)\chi^{*}(x)d^3x=0$
which leads to contradiction and (c) follows.

We shall also proof (d) by contradiction. Assume that there exists a
$\chi\in L^2$ such that
$\langle\nabla\chi,\mathds{1}_{B_\gamma}\nabla\chi\rangle+\langle\chi,
(v^N_{\beta}-W^N_{\beta_1,\beta_2})\chi\rangle<0$. Since our
potential is spherically symmetric we can assume without loss of
generality that $\chi$ is spherically symmetric. Defining the
function $ \varphi(r):=\chi(r)$ for $r\leq N^{-\gamma}$ and
$\varphi(r)=\chi(N^{-\gamma})$ for $r>N^{-\gamma}$ it follows that
$$\langle\varphi,h\varphi\rangle=
\langle\nabla\chi,\mathds{1}_{B_r}\nabla\chi\rangle+\langle\chi,
(v^N_{\beta}-W^N_{\beta_1,\beta_2})\chi\rangle<0\;.$$ This
contradicts (c) and (d) follows.

\end{proof}

\subsection{Control of the kinetic energy for $\beta=1$}

Next we shall control the kinetic energy $\|\nabla_1\Psi\|$ for
$\beta=1$. Note that in this case, a relevant part of the kinetic
energy is absorbed to form the microscopic structure. That part of
the kinetic energy is concentrated around the scattering centers.

The microscopic structure can --- as long as there are no three
particle interactions --- be controlled using Lemma \ref{defAlemma}.
So we shall first cutoff three particle interactions without
disturbing $\nabla_1\Psi$, i.e. we define a cutoff function which
does not depend on $x_1$ and cuts off all parts of the wave function
where two particles $x_j, x_k$ with $j\neq k$, $j,k\neq 1$ come to
close ($R_1$ given by Definition \ref{hdetail}).

After that we shall subtract that part of the kinetic energy which
is used to form the microscopic structure. The latter is
concentrated around the scattering centra (i.e. on the set
$\overline{\mathcal{S}}_j$ given by Definition \ref{hdetail}).

\begin{definition}\label{hdetail}
For any $j,k=\{1,\ldots,N\}$ let \be\label{defkleins}s_{j,k}:=\{X\in
\mathbb{R}^{3N}: |x_j-x_k|<N^{-26/27}\}\ee

$$\overline{\mathcal{S}}_j:=\bigcup_{k\neq j}s_{j,k}\;\;\;\;\;\;\;\mathcal{S}_j:=\mathbb{R}^{3N}\backslash \overline{\mathcal{S}}_j\;\;\;\;\;\;\;\mathcal{R}_{j,k}:=\bigcup_{l\neq
j,k}s_{k,l}\;\;\;\;\;\;\;\mathcal{R}_{j,k}:=\mathbb{R}^{3N}\backslash
\overline{\mathcal{R}}_{j,k}$$

\end{definition}

\begin{proposition}\label{propo}
$$\|\Psi_t-\mathds{1}_{R_j}\Psi_t\|<CN^{-7/54}\;.$$
\end{proposition}

\begin{proof}

Using H\"older and Sobolev we get 
\beas \|\Psi_t-\mathds{1}_{R_j}\Psi_t\|&=&\|\Psi_t
\mathds{1}_{\overline{\mathcal{R}}_{j,k}}\|^2
\leq\|\mathds{1}_{\overline{\mathcal{R}}_{j,k}}\|_{3/2}\;\|\Psi_t^2\|_3=|\overline{\mathcal{R}}_{j,k}|^{2/3}\|\Psi_t\|_6^2
\\
&\leq& (N-1)\|\nabla_1\Psi_t\|^2(N N^{-26/9})^{2/3}\leq
N^{-7/27}\|\nabla_1 \Psi_t\|^2\;.
 \eeas
Since $\|\nabla_1\Psi_t\|<C$ the Proposition follows.

\end{proof}

\begin{lemma}\label{kineticenergy}
Let under the conditions of the Theorem  $\beta\leq 1$. Then there
exists a $\gamma>0$ such that for any $t\in\mathbb{R}$
\be\label{kinen}\|\mathds{1}_{\mathcal{S}_{1}}\nabla_1q_1
\Psi_t\|<C(\alpha(\Psi_t)+\int_{0}^t \alpha(\Psi_s)ds+N^{-\gamma})
 \ee

\be\label{poten}\|\mathds{1}_{\overline{\mathcal{R}}_{1}}\sqrt{v}_N(x_1-x_2)\Psi_t\|<CN^{-1}(\alpha(\Psi_t)+\int_{0}^t
\alpha(\Psi_s)ds+N^{-\gamma}) \ee

\end{lemma}
\begin{proof} Below we shall use
from time to time that for any $f\in L^2$, any $g\in L^1$ and any
normalized $\Psi$, $\chi$
$$\|f(x_1-x_2)p_1\Psi\|^2=\langle\Psi p_1 f^2(x_1-x_2)p_1\Psi\rangle\leq \|f^2\|_1 \|\phi^{GP}\|_\infty^2$$
and $$\langle\chi p_1g(x_1-x_2)p_1\Psi\rangle\leq \|g\|_1
\|\phi^{GP}\|_\infty^2$$ Thus
\be\label{opnorm1}\|f(x_1-x_2)p_1\|_{op}\leq
\|f\|\|\phi^{GP}\|_\infty\ee and
\be\label{opnorm2}\|p_1g(x_1-x_2)p_1\|_{op}\leq\|g\|_1
\|\phi^{GP}\|_\infty^2\;,\ee where $\|\cdot\|_{op}$ stands for the
operator norm $$\|A\|_{op}:=\inf_{\|\Psi\|=1}\|A\Psi\|\;.$$

Let us now prove Lemma \ref{kineticenergy}. Recall (\ref{kinensum})
\beas &&C(N^{-\delta}+N^{-1/4}+\sup_{0\leq s\geq
t}\{\alpha(\Psi_s)\})\\&\geq&\|\nabla_1\Psi_t\|^2-\|\nabla\phi^{GP}_t\|^2+\left\langle\Psi_t,\left((N-1)v^N_{1}(x_1-x_2)-
a|\phi^{GP}_t|^2(x_1)\right)\Psi_t\right\rangle\nonumber
\\&=&\langle\nabla_1\Psi_t,\mathds{1}_{\mathcal{S}_{1}}\nabla_1\Psi_t\rangle
+\langle\nabla_1\Psi_t,\mathds{1}_{\overline{\mathcal{S}}_{1}}\nabla_1\Psi_t\rangle-E_{kin}^{GP}
\\&&+\langle\Psi_t,\sum_{j\neq
1}\mathds{1}_{\mathcal{R}_{1,j}}\left(v^N_{1}(x_1-x_j)-W^N_{\beta,1}(x_1-x_j)\right)\Psi_t\rangle
\\&&
+\langle\Psi_t,\left(\sum_{j\neq
1}\mathds{1}_{\mathcal{R}_{1,j}}W^N_{\beta,1}(x_1-x_j)-a|\phi^{GP}_t(x_1)|^2\right)\Psi_t\rangle
\\&&
+\langle\Psi_t,\sum_{j\neq
1}\mathds{1}_{\overline{\mathcal{R}}_{1,j}}v^N_{1}(x_1-x_j)\Psi_t\rangle\;.\eeas
By definition of the set $\mathcal{S}_{1}$ the support of the
potentials $v^N_{1}(x_1-x_j)$ and $W^N_{\beta,1}(x_1-x_j)$ are
subsets of
$\overline{\mathcal{S}}_{1}:=\mathbb{R}^{3N}\backslash\mathcal{S}_{1}$.
Furthermore we have by definition of the set $\mathcal{R}_{1,j}$
that the support of the potentials
$\mathds{1}_{\mathcal{R}_{1,j}}\left(v^N_{1}(x_1-x_j)-W^N_{\beta,1}(x_1-x_j)\right)$
are pairwise disjoint for different $j$. It follows with Lemma
\ref{defAlemma} (d) that
 \beas&&
\langle\nabla_1\Psi_t,\mathds{1}_{\overline{\mathcal{S}}_{1}}\nabla_1\Psi_t\rangle
+\langle\Psi_t,\sum_{j\neq
1}\mathds{1}_{\mathcal{R}_{1,j}}\left(v^N_{1}(x_1-x_j)-W^N_{\beta,1}(x_1-x_j)\right)\Psi_t\rangle
\eeas is positive and \bea\label{kineneabsch2}
&&\langle\nabla_1\Psi_t,\mathds{1}_{\overline{\mathcal{S}}_{1}}\nabla_1\Psi_t\rangle-E_{kin}^{GP}
+\langle\Psi_t,\left(\sum_{j\neq
1}\mathds{1}_{\mathcal{R}_{1,j}}W^N_{\beta,1}(x_1-x_j)-a|\phi^{GP}_t(x_1)|^2\right)\Psi_t\rangle
\nonumber\\&&
+\langle\Psi_t,\sum_{j\neq
1}\mathds{1}_{\overline{\mathcal{R}}_{1,j}}v^N_{1}(x_1-x_j)\Psi_t\rangle
\leq C(N^{-\delta}+N^{-1/4}+\sup_{0\leq s\leq t}\{\alpha(\Psi_s)\})
\;.\eea For the last summand in the first line we have using
positivity of $W^N_{\beta,1}$ \beas
&&\langle\Psi_t,\left(\sum_{j\neq
1}\mathds{1}_{\mathcal{R}_{1,j}}W^N_{\beta,1}(x_1-x_j)-a|\phi^{GP}_t(x_1)|^2\right)\Psi_t\rangle
\\&=&\langle\Psi_t,p_1p_2\left(\sum_{j\neq
1}W^N_{\beta,1}(x_1-x_j)-a|\phi^{GP}_t(x_1)|^2\right)p_1p_2\Psi_t\rangle
\\&&-\langle\Psi_t,p_1p_2\sum_{j\neq
1}\mathds{1}_{\overline{\mathcal{R}}_{1,j}}W^N_{\beta,1}(x_1-x_j)p_1p_2\Psi_t\rangle
\\&&+2\Re\left(\langle(1-p_1p_2)\Psi_t,\left(\sum_{j\neq
1}W^N_{\beta,1}(x_1-x_j)-a|\phi^{GP}_t(x_1)|^2\right)p_1p_2\Psi_t\rangle\right)
\\&&+2\Re\left(\langle(1-p_1p_2)\Psi_t,\sum_{j\neq
1}\mathds{1}_{\overline{\mathcal{R}}_{1,j}}W^N_{\beta,1}(x_1-x_j)p_1p_2\Psi_t\rangle\right)
\\&&+\langle(1-p_1p_2)\Psi_t,\sum_{j\neq
1}\mathds{1}_{\mathcal{R}_{1,j}}W^N_{\beta,1}(x_1-x_j)\langle(1-p_1p_2)\Psi_t\rangle
\\&&-\langle(1-p_1p_2)\Psi_t,a|\phi^{GP}_t(x_1)|^2\langle(1-p_1p_2)\Psi_t\rangle
\\&=:&\sum_{j=1}^6 S_j;
\eeas

We already got bounds on $S_1$, $S_3$ and $S_6$: All these terms
appeared in (\ref{positiverterm}) above and could be estimated by
the right hand side of Lemma \ref{ekin}. $S_5>0$ since
$W^N_{1,\beta}$ is positive. For $S_2$ we have \beas
&&\left|\langle\Psi_t,p_1p_2\sum_{j\neq
1}\mathds{1}_{\overline{\mathcal{R}}_{1,j}}W^N_{\beta,1}(x_1-x_j)p_1p_2\Psi_t\rangle\right|
\\&\leq&(N-1)\|p_1W^N_{\beta,1}(x_1-x_2)p_1\|_{op}\;\|\mathds{1}_{\overline{\mathcal{R}}_{1,j}}p_2\Psi_t\|^2
\\&\leq&(N-1)\|p_1W^N_{\beta,1}(x_1-x_2)p_1\|_{op}\left(\|\mathds{1}_{\overline{\mathcal{R}}_{1,j}}\Psi_t\|+\|\mathds{1}_{s_{2,j}}p_2\Psi_t\|\right)^2
\eeas which is in view of Proposition \ref{propo} bounded by the
right hand of (\ref{kinen}).

For $S_4$ we have \bea\label{s4est}
&&2\left|\Re\left(\langle(1-p_1p_2)\Psi_t,\sum_{j\neq
1}\mathds{1}_{\overline{\mathcal{R}}_{1,j}}W^N_{\beta,1}(x_1-x_j)p_1p_2\Psi_t\rangle\right)\right|
\\\nonumber&\leq&2(N-1)\|\sqrt{W^N_{\beta,1}}(x_1-x_2)(1-p_1p_2)\Psi_t\|\;\|\mathds{1}_{\overline{\mathcal{R}}_{1,2}}\sqrt{W^N_{\beta,1}}(x_1-x_2)p_1p_2\Psi_t\|
\eea Using H\"older and Sobolev we have
 \beas
 \|\sqrt{W^N_{\beta,1}}(x_1-x_2)\Psi_t\|^2
&\leq&\|W^N_{\beta,1}\|_{3/2}\;\|\Psi\|_6^2\\\leq \left(\int
 |W^N_{\beta,1}(x)|^{3/2}d^3x\right)^{2/3}\|\nabla\Psi\|^2
&\leq&C(N^{-3/2+3\beta/2})^{2/3}=CN^{-1+\beta}
 \;.
\eeas Since \beas
&&\|\mathds{1}_{\overline{\mathcal{R}}_{1,2}}\sqrt{W^N_{\beta,1}}(x_1-x_2)p_1p_2\Psi_t\|^2\leq
\|p_2\mathds{1}_{\overline{\mathcal{R}}_{1,2}}p_2\|_{op}\|\sqrt{W^N_{\beta,1}}(x_1-x_2)p_1\Psi_t\|^2
\\&\leq&|\overline{\mathcal{R}}_{1,2}|CN^{-1}\leq CNN^{-26/9}N^{-1}=
CN^{-26/9} \eeas it follows with (\ref{s4est}) that $S_4$ is bounded
by the right hand side of (\ref{kinen}). Hence
\bea\label{kineneabsch3}
&&\langle\nabla_1\Psi_t,\mathds{1}_{\overline{\mathcal{S}}_{1}}\nabla_1\Psi_t\rangle-E_{kin}^{GP}
+\langle\Psi_t,\sum_{j\neq
1}\mathds{1}_{\overline{\mathcal{R}}_{1,j}}v^N_{1}(x_1-x_j)\Psi_t\rangle
\nonumber\\&\leq& C(N^{-\delta}+N^{-1/4}+\sup_{0\leq s\leq
t}\{\alpha(\Psi_s)\}) \;.\eea

For the first summand in (\ref{kineneabsch3}) we can write \beas
\|\mathds{1}_{\mathcal{S}_{1}}\nabla_1\Psi_t\|^2&\geq&
\|\mathds{1}_{\mathcal{S}_{1}}\nabla_1
p_1\Psi_t\|^2+\|\mathds{1}_{\mathcal{S}_{1}}\nabla_1 q_1\Psi_t\|^2
-2|\langle\nabla_1q_1\Psi_t,\nabla_1p_1\Psi_t\rangle|
\\&&-2
|\langle\nabla_1q_1\Psi_t,\mathds{1}_{\overline{\mathcal{S}}_{1}}\nabla_1p_1\Psi_t\rangle|
\\&\geq&
\|\nabla_1 p_1\Psi_t\|^2+\|\mathds{1}_{\mathcal{S}_{1}}\nabla_1
q_1\Psi_t\|^2
-\|\mathds{1}_{\overline{\mathcal{S}}_{1}}\nabla_1 p_1\Psi_t\|^2
\\&&+\langle q_1\Psi_t,\Delta_1p_1\Psi_t\rangle
-
\|\nabla_1q_1\Psi_t\|\;\|\mathds{1}_{\overline{\mathcal{S}}_{1}}\|_1^{1/2}\;\|\nabla_1\phi^{GP}_t\|_\infty
\\&\geq&\|\nabla\phi^{GP}\|
\;\|p_1\Psi_t\|+\|\mathds{1}_{\mathcal{S}_{1}}\nabla_1
q_1\Psi_t\|^2\\&&-\|\mathds{1}_{\overline{\mathcal{S}}_{1}}\|^2_1\;\|\nabla_1\phi^{GP}_t\|^2_\infty
-\|q_1\Psi_t\|\;\|\Delta_1\phi^{GP}_t\|_\infty\;\|p_1\Psi_t\|
\\&&-
\left(\|\nabla_1\Psi_t\|+\|\nabla_1p_t\Psi_t\|\right)\;\|\mathds{1}_{\overline{\mathcal{S}}_{1}}\|_1^{1/2}\;\|\nabla_1\phi^{GP}_t\|_\infty\;.
\eeas Since
\be\label{1norms}\|\mathds{1}_{\overline{\mathcal{S}}_{1}}\|_1\leq N
|s_{1,2}|=4/3 \pi N^{-17/9}\ee we can find a $\gamma>0$ such that
\beas&&\|\mathds{1}_{\mathcal{S}_{1}}\nabla_1
q_1\Psi_t\|^2+\langle\Psi_t,\mathds{1}_{\overline{\mathcal{R}}_{1}}\sum_{j\neq
1}v^N_{1}(x_1-x_j)\Psi_t\rangle\\&\leq&
E_{kin}(1-\|p_t\Psi_t\|^2)+C(\alpha(\Psi_0)+\alpha(\Psi_t))+C\int_{0}^t
\alpha(\Psi_s)ds+CN^{-\gamma}\;.\eeas Using that
$1-\|p_t\Psi_t\|^2=\|q_t\Psi_t\|^2<\alpha(\Psi)$ and that both
summands are positive the Lemma follows.
\end{proof}

\subsection{Redefinition of $\alpha$ for $1/3<\beta\leq1$}

As mentioned in the introduction one has to control the microscopic
structure of $\Psi$ when $\beta$ increases. On the technical level
that means, that for $\beta>1/3$ the $\alpha_1'$ and for $\beta=1$
the $\alpha_2'$ can't be controlled. We have to equip the
$\alpha_{1/2}$ with the respective microscopic structure. We shall
do that by adding the functions $\lambda_{1,2}$ to $\alpha_{1/2}$
and $\lambda_{1,2}'$ to $\alpha_{1/2}'$ in such a way, that
$\lambda_{1,2}'(\Psi_t)$ is the time dericative
$\lambda_{1,2}(\Psi_t)$ if $\Psi_t$ solves the Schr\"odinger
equation and $\alpha_{1,2}+\lambda_{1,2}$ becomes controllable.

First note that we can replace in the estimate of the second term in
Lemma \ref{L2absch} (a) $\|\nabla_1q_1 \Psi_t\|$ by
$\|\mathds{1}_{\mathcal{S}_{1}}\nabla_1q_1 \Psi_t\|$:
\begin{lemma}\label{L2absch3}
Under the conditions of Lemma \ref{L2absch} we have for $0<\beta<1$
\beas&&|\langle\Psi_t ,p_1 q_2 \widehat{m}^{1/2} h_{1,2}
\widehat{m}^{1/2}q_1 q_2 \Psi_t\rangle|\\&&\hspace{2cm}\leq
C(\|\phi_t^{GP}\|_\infty+\|\nabla\phi_t^{GP}\|_\infty)
(\alpha(\Psi)+N^{-\gamma}+\|\mathds{1}_{\mathcal{S}_{1}}\nabla_1q_1
\Psi_t\|^2)\eeas
\end{lemma}
The proof shall be given in the Appendix.

\begin{definition}\label{lambda2}
Let  $v^N_{1}\in\mathcal{V}_1$.

We define \beas\lambda_{2}(\Psi) &:=&N(N-1)\Im\left(\left\langle\Psi
, g^N_{8/9,1}(x_{1}-x_{2}) (\widehat{n}-\widehat{n}_1)p_{1}q_{2}
\Psi\right\rangle\right) \eeas and \beas\lambda_{2}'(\Psi)
&:=&N(N-1)\Im\left(\left\langle\Psi ,\left[H,
g^N_{8/9,1}(x_{1}-x_{2})
(\widehat{n}-\widehat{n}_1)p_{1}q_{2}\right]
\Psi\right\rangle\right)
\\&&-N(N-1)\Im\left(\left\langle\Psi , g^N_{8/9,1}(x_{1}-x_{2})\left[H^{GP},
(\widehat{n}-\widehat{n}_1)p_{1}q_{2}\right] \Psi\right\rangle\right)
\;.\eeas

\end{definition}

\begin{lemma}\label{replacealpa2}
There exists a $\gamma>0$ such that
\begin{enumerate}
\item For any solution of the Sch\"odinger equation $\Psi_t\in\LZN$
$$i\frac{d}{dt}\lambda_{2}(\Psi_t)=\lambda_{2}'(\Psi_t)$$
\item
There exist a $C<\infty$ such that for any $\Psi\in\LZN$
\be\label{lambda-alpha}|\lambda_{2}'(\Psi)-\alpha'_{2}(\Psi)|\leq
C(\|\phi_t^{GP}\|_\infty+\|\nabla\phi_t^{GP}\|_\infty)\left(N^{-\gamma}+(\ln
N)^{1/3}\alpha(\Psi)\right)\;.\ee

\item $$|\lambda_{2}(\Psi)|\leq C N^{-\gamma}\|\phi_t^{GP}\|_\infty$$
\end{enumerate}
\end{lemma}

\begin{proof}
(a) follows as above, using that $\frac{d}{dt} \widehat{m}=-[H^{GP},\widehat{m}]$.\\
For (c) we have with Lemma \ref{kombinatorik} \beas
|\lambda_{2}(\Psi)|\leq N^2
\|g^N_{8/9,1}\|\;\|\phi^{GP}_t\|_\infty\;\|(\widehat{n}-\widehat{n}_1)q_2\Psi\|\leq
C N^{-4/9}\|\phi_t^{GP}\|_\infty \eeas

For (b) we have since
\beas&&[H,g^N_{8/9,1}(x_1-x_2)]=[H,f^N_{8/9,1}(x_1-x_2)]=-[\Delta_1+\Delta_2,f^N_{8/9,1}(x_1-x_2)]
\\&=&(\Delta_1+\Delta_2)f^N_{8/9,1}(x_1-x_2)+(\nabla_1f^N_{8/9,1}(x_1-x_2))\nabla_1+(\nabla_2f^N_{8/9,1}(x_1-x_2))\nabla_2
\\&=&(v^N_{\beta}-W^N_{8/9,1})f^N_{8/9,1}(x_1-x_2)+(\nabla_1g^N_{8/9,1}(x_1-x_2))\nabla_1+(\nabla_2g^N_{8/9,1}(x_1-x_2))\nabla_2
\eeas that \beas
\lambda_{2}'(\Psi)&=&N(N-1)\Im\left(\left\langle\Psi
,\left[(H-H^{GP}), g^N_{8/9,1}(x_{1}-x_{2})
(\widehat{n}-\widehat{n}_1)p_{1}q_{2}\right]
\Psi\right\rangle\right)
\\&&+N(N-1)\Im\big(\langle\Psi ,\big((v^N_{\beta}-W^N_{8/9,1})f^N_{8/9,1}(x_1-x_2)\\&&+(\nabla_1g^N_{8/9,1}(x_1-x_2))\nabla_1+(\nabla_2g^N_{8/9,1}(x_1-x_2))\nabla_2\big)
(\widehat{n}-\widehat{n}_1)p_{1}q_{2}\Psi\rangle\big)\;.
\eeas It follows  that \beas
\lambda_{2}'(\Psi)-\alpha'_{2}(\Psi)&=&N(N-1)\Im\left(\left\langle\Psi
,\left[-a_N\sum_{j=1}^N|\phi^{GP}_t|^2(x_j),
g^N_{8/9,1}(x_{1}-x_{2})
(\widehat{n}-\widehat{n}_1)p_{1}q_{2}\right]
\Psi\right\rangle\right)\\&&+
N(N-1)\Im\left(\left\langle\Psi ,W^N_{8/9,1}g^N_{8/9,1}(x_1-x_2)
(\widehat{n}-\widehat{n}_1)p_{1}q_{2}\Psi\right\rangle\right)
\\&&+N(N-1)\Im\big(\langle\Psi ,\big((\nabla_1g^N_{8/9,1}(x_1-x_2))\nabla_1\\&&+(\nabla_2g^N_{8/9,1}(x_1-x_2))\nabla_2\big)
(\widehat{n}-\widehat{n}_1)p_{1}q_{2}\Psi\rangle\big)
\\&&-N(N-1)\Im\left(\left\langle\Psi ,g^N_{8/9,1}(x_{1}-x_{2})
(\widehat{n}-\widehat{n}_1)p_{1}q_{2}\sum_{j< k}v^N_{\beta}(x_j-x_k)
\Psi\right\rangle\right)
\\&&+N(N-1)\Im\left(\left\langle\Psi ,\sum_{j=1}^2\sum_{k=3}^Nv^N_{\beta}(x_j-x_k) g^N_{8/9,1}(x_{1}-x_{2})
(\widehat{n}-\widehat{n}_1)p_{1}q_{2} \Psi\right\rangle\right)
\\&&+N(N-1)\Im\left(\left\langle\Psi ,\sum_{2<j<k}^Nv^N_{\beta}(x_j-x_k) g^N_{8/9,1}(x_{1}-x_{2})
(\widehat{n}-\widehat{n}_1)p_{1}q_{2} \Psi\right\rangle\right)
\;.
\eeas Using symmetry of $\Psi$ and
$\nabla_1g^N_{\beta_1,\beta_2}=-\nabla_2g^N_{\beta_1,\beta_2}$
\bea\label{sobis10}
&&\lambda_{2}'(\Psi)-\alpha'_{2}(\Psi)\nonumber\\&=&N(N-1)\Im\left(\left\langle\Psi
,\left[-a_N\sum_{j=1}^N|\phi^{GP}_t|^2(x_j),
g^N_{8/9,1}(x_{1}-x_{2})
(\widehat{n}-\widehat{n}_1)p_{1}q_{2}\right]
\Psi\right\rangle\right)\nonumber\\&&+
N(N-1)\Im\left(\left\langle\Psi ,W^N_{8/9,1}f^N_{8/9,1}(x_1-x_2)
(\widehat{n}-\widehat{n}_1)p_{1}q_{2}\Psi\right\rangle\right)
\nonumber\\&&-N(N-1)\Im\left(\left\langle\Psi
,(\nabla_2g^N_{8/9,1}(x_1-x_2))\nabla_1
p_{1}q_{2}(\widehat{n}-\widehat{n}_1)\Psi\right\rangle\right)
\nonumber\\&&-N(N-1)\Im\left(\left\langle\Psi
,(\nabla_1g^N_{8/9,1}(x_1-x_2))\nabla_2
p_{1}q_{2}(\widehat{n}-\widehat{n}_1)\Psi\right\rangle\right)
\nonumber\\&&-N(N-1)\Im\left(\left\langle\Psi
,g^N_{8/9,1}(x_{1}-x_{2})
(\widehat{n}-\widehat{n}_1)p_{1}q_{2}v^N_{\beta}(x_1-x_2)
\Psi\right\rangle\right)
\nonumber\\&&-\frac{N!}{(N-3)!}\Im\left(\left\langle\Psi
,g^N_{8/9,1}(x_{1}-x_{2})
(\widehat{n}-\widehat{n}_1)p_{1}q_{2}v^N_{\beta}(x_2-x_3)
\Psi\right\rangle\right)
\nonumber\\&&-\frac{N!}{(N-3)!}\Im\left(\left\langle\Psi
,g^N_{8/9,1}(x_{1}-x_{2})
(\widehat{n}-\widehat{n}_1)p_{1}q_{2}v^N_{\beta}(x_1-x_3)
\Psi\right\rangle\right)
\nonumber\\&&-\frac{N!}{(N-4)!}\Im\left(\left\langle\Psi
,g^N_{8/9,1}(x_{1}-x_{2})
(\widehat{n}-\widehat{n}_1)p_{1}q_{2}v^N_{\beta}(x_3-x_4)
\Psi\right\rangle\right)
\nonumber\\&&+\frac{N!}{(N-3)!}\Im\left(\left\langle\Psi
,v^N_{\beta}(x_2-x_3) g^N_{8/9,1}(x_{1}-x_{2})
(\widehat{n}-\widehat{n}_1)p_{1}q_{2} \Psi\right\rangle\right)
\nonumber\\&&+\frac{N!}{(N-3)!}\Im\left(\left\langle\Psi
,v^N_{\beta}(x_1-x_3) g^N_{8/9,1}(x_{1}-x_{2})
(\widehat{n}-\widehat{n}_1)p_{1}q_{2} \Psi\right\rangle\right)
\nonumber\\&&+\frac{N!}{(N-4)!}\Im\left(\left\langle\Psi
,v^N_{\beta}(x_3-x_4) g^N_{8/9,1}(x_{1}-x_{2})
(\widehat{n}-\widehat{n}_1)p_{1}q_{2} \Psi\right\rangle\right)
\nonumber\\&=:&\sum_{j=0}^{10} S_j
\;.
\eea For the first summand we have \beas
|S_0|&\leq&2N^2a\|\phi^{GP}_t\|_\infty^3\; \|g^N_{8/9,1}\|\;
\|(\widehat{n}-\widehat{n}_1)p_{1}q_{2} \Psi\|\;.
\eeas
With Lemma \ref{kineticenergy} it follows that $|S_0|$ is bounded by the right hand side of (\ref{lambda-alpha}).

Using as above (see proof of Lemma \ref{alphaabl}) that
$\Im(\langle\Psi,A\Psi\rangle)=-\Im(\langle\Psi,A^t\Psi\rangle)$ for
any operator $A$ and that $\Psi$ is symmetric (note that $p_1
q_2v^N_{\beta}(x_1-x_2)q_1  p_2$ is invariant under adjunction plus
exchange of the variable $x_1$ and $x_2$) and Lemma
\ref{kombinatorik} (d) we get for $S_1$ \beas
|S_1|&\leq&N(N-1)\left|\Im\left(\left\langle\Psi
,p_1p_2W^N_{8/9,1}f^N_{8/9,1}(x_1-x_2)
(\widehat{n}-\widehat{n}_1)p_{1}q_{2}\Psi\right\rangle\right)\right|
\\&&+N(N-1)\left|\Im\left(\left\langle\Psi ,q_1q_2W^N_{8/9,1}f^N_{8/9,1}(x_1-x_2)
(\widehat{n}-\widehat{n}_1)p_{1}q_{2}\Psi\right\rangle\right)\right|
\eeas Since $W^N_{8/9,1}f^N_{8/9,1}\in\mathcal{V}_{\beta_2}$ (see
Lemma \ref{defAlemma} (b)) it follows with Lemma \ref{L2absch3} that
$$|S_1|\leq
C(\|\phi_t^{GP}\|_\infty+\|\nabla\phi_t^{GP}\|_\infty)(\alpha(\Psi)+\|\mathds{1}_{\mathcal{S}_{1}}\nabla_1q_1
\Psi_t\|+N^{-\xi})\;.$$ With Lemma \ref{kineticenergy} it follows
that $|S_1|$ is bounded by the right hand side of
(\ref{lambda-alpha}).

For $S_2$ and $S_3$ we get integrating by parts \beas
S_2+S_3&=&N(N-1)\Im\left(\left\langle\nabla_2\Psi
,(g^N_{8/9,1}(x_1-x_2))\nabla_1
p_{1}q_{2}(\widehat{n}-\widehat{n}_1)\Psi\right\rangle\right)
\\&&\hspace{-0,3cm}N(N-1)\Im\left(\left\langle(q_1q_2+1-q_1q_2)\nabla_1\Psi ,(g^N_{8/9,1}(x_1-x_2))\nabla_2
p_{1}q_{2}(\widehat{n}-\widehat{n}_1)\Psi\right\rangle\right)
\\&&\hspace{-0,3cm}+2N(N-1)\Im\left(\left\langle(q_1q_2+1-q_1q_2)\Psi ,(g^N_{8/9,1}(x_1-x_2))\nabla_1\nabla_2
p_{1}q_{2}(\widehat{n}-\widehat{n}_1)\Psi\right\rangle\right) \;.
\eeas Thus \beas
|S_2+S_3|&\leq&N(N-1)\|\nabla_2\Psi\|\;\|g^N_{8/9,1}\|\;\|\nabla_1
\phi^{GP}_t\|_{\infty}\;\|p_{1}q_{2}(\widehat{n}-\widehat{n}_1)\Psi\|
\\&&\hspace{-0,3cm}+N(N-1)\|q_1\nabla_2q_2(\widehat{n}_2-\widehat{n}_3)\Psi\|\;\|g^N_{8/9,1}\|\;\|
\phi^{GP}_t\|_{\infty}\;\|\nabla_2p_{1}q_{2}(\widehat{n}-\widehat{n}_1)\Psi\|
\\&&\hspace{-0,3cm}+N(N-1)\|\nabla_1\Psi\|\;\|g^N_{8/9,1}\|_1\;\|
\phi^{GP}_t\|_{\infty}^2\;\|\widehat{n}-\widehat{n}_1\|_{op}\;\|\nabla_2q_{2}\Psi\|
\\&&\hspace{-0,3cm}+2N(N-1)\|q_1q_2(\widehat{n}_2-\widehat{n}_3)\Psi\|\;\|g^N_{8/9,1}\|_1\;\|
\phi^{GP}_t\|_{\infty}\;\|
\nabla_1\phi^{GP}_t\|_{\infty}\;\|\nabla_2p_{1}q_{2}(\widehat{n}-\widehat{n}_1)\Psi\|
\\&&\hspace{-0,3cm}+2N(N-1)\|\Psi\|\;\|g^N_{8/9,1}\|_1\;\|
\phi^{GP}_t\|_{\infty}\;\|
\nabla_1\phi^{GP}_t\|_{\infty}\;\|\widehat{n}-\widehat{n}_1\|_{op}\;\|\nabla_2q_{2}\Psi\|
\;.
\eeas
With Lemma \ref{kineticenergy} it follows that $|S_2+S_3|$ is bounded by the right hand side of (\ref{lambda-alpha}).

For $S_4$ we have
$$|S_4|\leq N(N-1)\|\Psi\|\;\|g^N_{8/9,1}(x_1-x_2)\|\;\|\phi^{GP}_t\|_\infty\;\|\|\sqrt{v^N_{\beta}}\|\;\|\sqrt{v^N_{\beta}}\Psi\|+a\|\phi^{GP}\|_\infty\;\|\Psi\|)\;.$$
With Lemma \ref{kineticenergy} it follows that $|S_4|$ is bounded by the right hand side of (\ref{lambda-alpha}).

For $S_5$  we have using $q_2=1-p_2$\beas
|S_5|&\leq&\frac{N!}{(N-3)!}\left|\Im\left(\left\langle(p_1+q_1)\sqrt{v}^N_{\beta}(x_2-x_3)\Psi
,g^N_{8/9,1}(x_{1}-x_{2})
(\widehat{n}-\widehat{n}_1)p_{1}\sqrt{v}^N_{\beta}(x_2-x_3)
\Psi\right\rangle\right)\right|
\\&&+\frac{N!}{(N-3)!}\left|\Im\left(\left\langle\Psi
,(q_1q_2+1-q_1q_2)g^N_{8/9,1}(x_{1}-x_{2})
(\widehat{n}-\widehat{n}_1)p_{1}p_{2}v^N_{\beta}(x_2-x_3)
\Psi\right\rangle\right)\right|
\eeas For the first summand we have \beas&&
\frac{N!}{(N-3)!}\left|\Im\left(\left\langle(p_1+q_1)\sqrt{v}^N_{\beta}(x_2-x_3)\Psi
,g^N_{8/9,1}(x_{1}-x_{2})
(\widehat{n}-\widehat{n}_1)p_{1}\sqrt{v}^N_{\beta}(x_2-x_3)
\Psi\right\rangle\right)\right|
\\&\leq&\frac{N!}{(N-3)!}\|\sqrt{v^N_{\beta}(x_2-x_3)}\Psi\|\;\|(\widehat{n}_1-\widehat{n}_2)q_1q_2\|_{op}
\|g^N_{8/9,1}\|_1\;\|\phi^{GP}_t\|_\infty^2\;\|\sqrt{v^N_{\beta}(x_2-x_3)}\Psi\|
\\&&+\frac{N!}{(N-3)!}\|(\widehat{n}_1-\widehat{n}_2)q_1\sqrt{v^N_{\beta}(x_2-x_3)}\Psi\|
\\&&\hspace{3cm}\|g^N_{8/9,1}\|\;\|\phi^{GP}_t\|_\infty\;\|\sqrt{v^N_{\beta}(x_2-x_3)}\Psi\|
\eeas which is due to Lemma \ref{kineticenergy} and Lemma
\ref{kombinatorik} bounded by the right hand side of
(\ref{lambda-alpha}).
 For the second summand we have in view of Lemma
\ref{kombinatorik} (d)
\beas&&\frac{N!}{(N-3)!}\left|\Im\left(\left\langle\Psi
,(q_1q_2+1-q_1q_2)g^N_{8/9,1}(x_{1}-x_{2})
(\widehat{n}-\widehat{n}_1)p_{1}p_{2}v^N_{\beta}(x_2-x_3)
\Psi\right\rangle\right)\right|
\\&\leq
&\frac{N!}{(N-3)!}\|(\widehat{n}_1-\widehat{n}_2)q_1q_2\Psi\|\;\|g^N_{8/9,1}\|\;\|\phi^{GP}_t\|_\infty^2\;\|\sqrt{v^N_{\beta}}\|\;\|\sqrt{v^N_{\beta}(x_2-x_3)}\Psi\|
\\&&+\frac{N!}{(N-3)!}\|g^N_{8/9,1}\|_1\;\|\phi^{GP}_t\|^2\;\|\widehat{n}-\widehat{n}_1\|_{op}
\;\|\phi^{GP}_t\|_\infty\;\|\sqrt{v^N_{\beta}}\|\;\|\sqrt{v^N_{\beta}(x_2-x_3)}\Psi\|
\;.\eeas With Lemma
\ref{kineticenergy} it follows that $|S_5|$ is bounded by the right
hand side of (\ref{lambda-alpha}).

Similarly we get for $S_6$ using Lemma \ref{kombinatorik} (d) \beas
|S_6|&=&\frac{N!}{(N-3)!}\left|\Im\left(\left\langle\Psi
,(q_1q_2+1-q_1q_2)g^N_{8/9,1}(x_{1}-x_{2})
(\widehat{n}-\widehat{n}_1)p_{1}q_{2}v^N_{\beta}(x_1-x_3)
\Psi\right\rangle\right)\right|
\\&\leq &\frac{N!}{(N-3)!}\|(\widehat{n}_1-\widehat{n}_2)q_1q_2\Psi\|\;\|g^N_{8/9,1}\|\;\|\phi^{GP}_t\|_\infty^2\;\|\sqrt{v^N_{\beta}(x_1-x_3)}\|\;\|\sqrt{v^N_{\beta}(x_1-x_3)}\Psi\|
\\&&+\|(\widehat{n}_1-\widehat{n}_2)q_1q_2\Psi\|\;\|g^N_{8/9,1}\|\;\|\phi^{GP}_t\|_\infty\;a/N\|\phi^{GP}_t\|\;\|\Psi\|
\\&&+\|g^N_{8/9,1}\|_1\;\|\phi^{GP}_t\|^2\;\|\widehat{n}-\widehat{n}_1\|_{op}\;
%
\|\phi^{GP}_t\|_\infty\;\|\sqrt{v^N_{\beta}}\|\;\|\sqrt{v^N_{\beta}}(x_1-x_3)\Psi
\|\;.\eeas With Lemma \ref{kineticenergy} it follows that $|S_6|$ is
bounded by the right hand side of (\ref{lambda-alpha}).

For $S_7+S_{10}$ we use (\ref{nersetzen})
to get \beas
S_7+S_{10}&=&-\frac{N!}{(N-4)!}\Im\left(\left\langle\Psi
,g^N_{8/9,1}(x_{1}-x_{2})Q v^N_{\beta}(x_3-x_4)
\Psi\right\rangle\right)
\\&&+\frac{N!}{(N-4)!}\Im\left(\left\langle\Psi ,g^N_{8/9,1}(x_{1}-x_{2})v^N_{\beta}(x_3-x_4)
Q\Psi\right\rangle\right)
\eeas with
$$Q=(\widehat{n}-\widehat{n}_2-\widehat{n}_1+\widehat{n}_3)p_{1}q_{2}p_3p_4+(\widehat{n}-2\widehat{n}_1+\widehat{n}_2)(p_{1}q_{2}p_3q_4+p_{1}q_{2}q_3p_4)\;.$$

Since $\sqrt{k}-\sqrt{k-2}-\sqrt{k-1}+\sqrt{k-3}<Ck^{-3/2}$ and
$\sqrt{k}-2\sqrt{k-1}+\sqrt{k-2}<Ck^{-3/2}$ it follows that
$$Q<N^{-2}
\widehat{n}^{-3/2}(p_{1}q_{2}p_3p_4+p_{1}q_{2}p_3q_4+p_{1}q_{2}q_3p_4)\;.$$
It follows using symmetry and Lemma \ref{kombinatorik} that \beas
&&|S_7+S_{10}|\\&\leq&N^2\left|\Im\left(\left\langle\Psi
,q_1q_2\widehat{n}^{-1}_1g^N_{8/9,1}(x_{1}-x_{2})\widehat{n}^{-1/2}(p_{1}q_{2}p_3p_4+2p_{1}q_{2}p_3q_4)
v^N_{\beta}(x_3-x_4) \Psi\right\rangle\right)\right|
\\&&+N^2\left|\Im\left(\left\langle\Psi ,(1-q_1q_2)g^N_{8/9,1}(x_{1}-x_{2})\widehat{n}^{-3/2}(p_{1}q_{2}p_3p_4+2p_{1}q_{2}p_3q_4)
v^N_{\beta}(x_3-x_4) \Psi\right\rangle\right)\right|
\\&&+N^2\left|\Im\left(\left\langle\Psi ,q_1q_2v^N_{\beta}(x_3-x_4)\widehat{n}^{-1}_1g^N_{8/9,1}(x_{1}-x_{2})\widehat{n}^{-1/2}(p_{1}q_{2}p_3p_4+2p_{1}q_{2}p_3q_4)
\Psi\right\rangle\right)\right|
\\&&+N^2\left|\Im\left(\left\langle\Psi ,(1-q_1q_2)g^N_{8/9,1}(x_{1}-x_{2})v^N_{\beta}(x_3-x_4)\widehat{n}^{-3/2}(p_{1}q_{2}p_3p_4+2p_{1}q_{2}p_3q_4)
 \Psi\right\rangle\right)\right|
 \\&\leq&3N^2\|\widehat{n}^{-1}_1q_1q_2\Psi \|\;\|g^N_{8/9,1}\|\;\|\phi^{GP}_t\|_\infty^2\;\|\sqrt{v^N_{\beta}}\|\;\|\sqrt{v^N_{\beta}}\Psi\|
 \\&&+3N^2\|g^N_{8/9,1}\|_1\;\|\phi^{GP}_t\|_\infty^3\;\|\sqrt{v^N_{\beta}}\|\;\|\widehat{n}^{-3/2}q_{2}\sqrt{v^N_{\beta}}(x_3-x_4)\Psi\|
 \\&&+3N^2\sum_{j=1}^3\|\sqrt{v^N_{\beta}}\widehat{n}^{-1}_jq_1q_2\Psi \|\;\|g^N_{8/9,1}\|\;\|\phi^{GP}_t\|_\infty^2\;\|\sqrt{v^N_{\beta}}\|
\\&&+3N^2\|\sqrt{v^N_{\beta}}(x_3-x_4)\Psi\|\;\|g^N_{8/9,1}\|_1\;\|\phi^{GP}_t\|_\infty^3\;\|\sqrt{v^N_{\beta}}\|\;\|\widehat{n}^{-3/2}q_{2}\Psi\|\;.
\eeas
With Lemma \ref{kineticenergy} it follows that $|S_7+S_{10}|$ is bounded by the right hand side of (\ref{lambda-alpha}).

For $S_8$ and $S_9$ note first, that

\beas
&&\left|\frac{N!}{(N-3)!}\Im\left(\left\langle\Psi
,v^N_{\beta}(x_2-x_3) g^N_{8/9,1}(x_{1}-x_{2})
(\widehat{n}-\widehat{n}_1)p_{1}p_{2} \Psi\right\rangle\right)\right|%
\\\nonumber&\leq&CN^2\left|\left\langle\Psi ,p_1v^N_{\beta}(x_2-x_3) g^N_{8/9,1}(x_{1}-x_{2})
\widehat{n}^{-1/2}p_{1}p_{2} \Psi\right\rangle\right|
\\\nonumber&&+CN^2\left|\left\langle\Psi ,q_1v^N_{\beta}(x_2-x_3) g^N_{8/9,1}(x_{1}-x_{2})
\widehat{n}^{-1/2}p_{1}p_{2} \Psi\right\rangle\right|
\\\nonumber&\leq&CN^2\|\sqrt{v^N_{\beta}}(x_2-x_3)\Psi\|\;\|\phi^{GP}_t\|^2\; \|g^N_{8/9,1}\|_1 N^{1/2} \|\phi^{GP}_t\|\;\|\sqrt{v^N_{\beta}}(x_2-x_3)\|
\\\nonumber&&+\sum_{j=1}^3 CN^2\|\widehat{n}^{-1/2}_jv^N_{\beta}(x_2-x_3) q_1\Psi\|\;\|\phi^{GP}_t\|\; \|g^N_{8/9,1}\|\; \|\phi^{GP}_t\|\;\|\sqrt{v^N_{\beta}}(x_2-x_3)\|
\eeas
With Lemma \ref{kineticenergy} it follows that the latter is bounded by the right hand side of (\ref{lambda-alpha}), thus it suffices to control
$$\widetilde{S}_8:=(N-1)(N-2)\Im\left(\left\langle\Psi ,v^N_{\beta}(x_2-x_3) g^N_{8/9,1}(x_{1}-x_{2})
(\widehat{n}-\widehat{n}_1)p_{1} \Psi\right\rangle\right)$$ and
$$\widetilde{S}_9:=\frac{N!}{(N-3)!}\Im\left(\left\langle\Psi ,v^N_{\beta}(x_1-x_3) g^N_{8/9,1}(x_{1}-x_{2})
(\widehat{n}-\widehat{n}_1)p_{1} \Psi\right\rangle\right)$$ instead
of  $S_8$ and $S_9$. For $\widetilde{S}_8$ we have \beas
|\widetilde{S}_8|&\leq&N^2\|\sqrt{v^N_{\beta}}(x_2-x_3)\Psi\|\;\|\phi^{GP}_t\|\;
\|g^N_{8/9,1}\|\;\|\sqrt{v^N_{\beta}}(x_2-x_3)\Psi\| \eeas which is
again bounded by the  right hand side of (\ref{lambda-alpha}). For
$\widetilde{S}_9$ \beas
|\widetilde{S}_9|&\leq&N^2\|\sqrt{v^N_{\beta}}(x_1-x_3)\Psi\|\;\|g^N_{8/9,1}(x_1-x_2)\sqrt{v^N_{\beta}}(x_1-x_3)p_1\Psi\|\;.
\eeas Note that due to (\ref{gbound})
$\sqrt{v}_{1,3}g^N_{8/9,1}(x_1-x_2)<v_{1,3}a/(N |x_1-x_2|)$ and
$g^N_{8/9,1}(x_1-x_2)<C$, thus
$\sqrt{v}_{1,3}g^N_{8/9,1}(x_1-x_2)<v_{1,3}\widetilde{g}(x_2-x_3)$
with $\widetilde{g}(x)<C/(N |x|+1)$ and $g(x)=0$ for $x>CN^{-8/9}$.
It follows that \beas
\|g^N_{8/9,1}(x_1-x_2)\sqrt{v^N_{\beta}}(x_1-x_3)p_1\Psi\|&\leq&C\|\phi^{GP}_t\|_\infty\;\|\sqrt{v^N_{\beta}}(x_1-x_3)\|\;\|
\widetilde{g}(x_2-x_3)\Psi\| \eeas Using H\"older and Sobolev it
follows that for sufficiently large $N$ \beas\|
\widetilde{g}(x_2-x_3)\Psi\|^2&\leq&\|\widetilde{g}^2\|_{3/2}\;\|\Psi^2\|_3=\|\widetilde{g}\|^2_{3}\;\|\Psi\|^2_6
\\&\leq&\|\nabla\Psi\|^2\;\left(\int\widetilde{g}^3d^3x\right)^{2/3}
\\&\leq&C\|\nabla\Psi\|^2\;\left(N^{-3}\int_{N^{-1}<|x|<1}|x|^{-3}d^3x+N^{-3}\right)^{2/3}
\\&\leq&CN^{-2}\|\nabla\Psi\|^2(\ln N)^{2/3}
\eeas
It follows that also $\widetilde{S}_9$ is bounded by the  right hand side of (\ref{lambda-alpha}) and (b) follows.

(c)
\end{proof}

Similar as for $\alpha_{2}(\Psi)$ above, we wish to equip
$\alpha_{1}(\Psi)$ with a microscopic structure, i.e. define a
$\lambda_{1}(\Psi)$ and a $\lambda_{1}'(\Psi)$ such that
$\frac{d}{dt}\lambda_{1}(\Psi)=\lambda_{1}'(\Psi)$ for any solution
of the Schr\"odinger equation $\Psi_t$ and such that
$\alpha_{1}'(\Psi)-\lambda_{1}'(\Psi)$ and $\lambda_{1}(\Psi)$
become controllable for $1/3\leq\beta\leq1$. As a first step we
shall define $\lambda_{1}(\Psi)$ similar as $\lambda_{2}(\Psi)$
above (c.f. Definition Lemma \ref{lambda2}) comparing
$\alpha_{1}(\Psi)$ with $\alpha_{1}(\Psi)$, i.e. we define \beas
\lambda_{1}(\Psi):=N(N-1)\Im\left(\left\langle\Psi ,
g^N_{2/7,\beta}(x_{1}-x_{2}) (\widehat{n}-\widehat{n}_2)p_{1}p_{2}
\Psi\right\rangle\right) \eeas and \beas\lambda_{1}'(\Psi)
&:=&N(N-1)\Im\left(\left\langle\Psi ,\left[H,
g^N_{2/7,\beta}(x_{1}-x_{2})
(\widehat{n}-\widehat{n}_2)p_{1}p_{2}\right]
\Psi\right\rangle\right)
\\&&-N(N-1)\Im\left(\left\langle\Psi , g^N_{2/7,\beta}(x_{1}-x_{2})\left[H^{GP},
(\widehat{n}-\widehat{n}_2)p_{1}p_{2}\right] \Psi\right\rangle\right)
\;.\eeas

As above (Lemma \ref{replacealpa2} (a) and (c)) we have
$$i\frac{d}{dt}\lambda_{1}(\Psi_t)=\lambda_{1}'(\Psi_t)\;.$$ Writing
\beas \lambda_{1}(\Psi):=N(N-1)\Im\left(\left\langle(p_1+q_1)\Psi ,
g^N_{2/7,\beta}(x_{1}-x_{2}) (\widehat{n}-\widehat{n}_2)p_{1}p_{2}
\Psi\right\rangle\right) \eeas we get furthermore
\beas\|\lambda_{1}(\Psi)\|&\leq& C
N^2\|\phi^{GP}_t\|^2_\infty\;\|g^N_{2/7,\beta}\|_1\;
N^{-1/2}+CN^2\|\phi^{GP}_t\|_\infty\;\|g^N_{2/7,\beta}\|\;N^{-1}
\\&\leq& C\|\phi^{GP}_t\|_\infty N^{-1/14} \;.\eeas

For $|\lambda_{1}'(\Psi)-\alpha'_{1}(\Psi)|$ we can use
(\ref{sobis10}), replacing $g^N_{8/9,\beta}$ by $g^N_{2/7,\beta}$
and $(\widehat{n}-\widehat{n}_1)p_{1}q_{2}$ by
$(\widehat{n}-\widehat{n}_2)p_{1}p_{2}$. Using symmetry, $1=p_j+q_j$
and (\ref{nersetzen}) and reordering the summands we get
\bea\label{sobis10b}\nonumber
&&\lambda_{1}'(\Psi)-\alpha'_{1}(\Psi)\\\nonumber&=&
N(N-1)\Im\left(\left\langle p_1\Psi
,\left[-a_N\sum_{j=1}^N|\phi^{GP}_t|^2(x_j),
g^N_{2/7,\beta}(x_{1}-x_{2})
(\widehat{n}-\widehat{n}_2)p_{1}p_{2}\right]
\Psi\right\rangle\right)
\\&&+
N(N-1)\Im\left(\left\langle q_1\Psi
,\left[-a_N\sum_{j=1}^N|\phi^{GP}_t|^2(x_j),
g^N_{2/7,\beta}(x_{1}-x_{2})
(\widehat{n}-\widehat{n}_2)p_{1}p_{2}\right]
\Psi\right\rangle\right)
\nonumber\\&&+N(N-1)\Im\left(\left\langle\Psi
,W^N_{2/7,\beta}f^N_{2/7,\beta}(x_1-x_2)
(\widehat{n}-\widehat{n}_2)p_{1}p_{2}\Psi\right\rangle\right)
\nonumber\\&&-2N(N-1)\Im\left(\left\langle\Psi
,p_1(\nabla_2g^N_{2/7,\beta}(x_1-x_2))\nabla_1
p_{1}p_{2}(\widehat{n}-\widehat{n}_1)\Psi\right\rangle\right)
\nonumber\\&&-2N(N-1)\Im\left(\left\langle\Psi
,q_1(\nabla_2g^N_{2/7,\beta}(x_1-x_2))\nabla_1
p_{1}p_{2}(\widehat{n}-\widehat{n}_1)\Psi\right\rangle\right)
\nonumber\\&&-N(N-1)\Im\left(\left\langle\Psi
,g^N_{2/7,\beta}(x_{1}-x_{2})
(\widehat{n}-\widehat{n}_2)p_{1}p_{2}v^N_{\beta}(x_1-x_2)
\Psi\right\rangle\right)
\nonumber\\&&-2\frac{N!}{(N-3)!}\Im\left(\left\langle\Psi
,p_1g^N_{2/7,\beta}(x_{1}-x_{2})
(\widehat{n}-\widehat{n}_2)p_{1}p_{2}v^N_{\beta}(x_2-x_3)
\Psi\right\rangle\right)
\nonumber\\&&-2\frac{N!}{(N-3)!}\Im\left(\left\langle\Psi
,q_1g^N_{2/7,\beta}(x_{1}-x_{2})
(\widehat{n}-\widehat{n}_2)p_{1}p_{2}v^N_{\beta}(x_2-x_3)
\Psi\right\rangle\right)
\nonumber\\&&+2\frac{N!}{(N-3)!}\Im\left(\left\langle\Psi
,p_1v^N_{\beta}(x_2-x_3) g^N_{2/7,\beta}(x_{1}-x_{2})
(\widehat{n}-\widehat{n}_2)p_{1}p_{2} \Psi\right\rangle\right)
\nonumber\\&&+2\frac{N!}{(N-3)!}\Im\left(\left\langle\Psi
,q_1v^N_{\beta}(x_2-x_3) g^N_{2/7,\beta}(x_{1}-x_{2})
(\widehat{n}-\widehat{n}_2)p_{1}p_{2} \Psi\right\rangle\right)
\nonumber\\&&+2\frac{N!}{(N-4)!}\Im\left(\left\langle\Psi
,v^N_{\beta}(x_3-x_4) g^N_{2/7,\beta}(x_{1}-x_{2})
(\widehat{n}-\widehat{n}_1-\widehat{n}_2+\widehat{n}_3)p_{1}p_{2}p_3q_3
\Psi\right\rangle\right)
\nonumber\\&&-2\frac{N!}{(N-4)!}\Im\left(\left\langle\Psi
,g^N_{2/7,\beta}(x_{1}-x_{2})
(\widehat{n}-\widehat{n}_1-\widehat{n}_2+\widehat{n}_3)p_{1}p_{2}p_3q_3v^N_{\beta}(x_3-x_4)
\Psi\right\rangle\right)
\nonumber\\&&+\frac{N!}{(N-4)!}\Im\left(\left\langle\Psi
,v^N_{\beta}(x_3-x_4) g^N_{2/7,\beta}(x_{1}-x_{2})
(\widehat{n}-2\widehat{n}_2+\widehat{n}_4)p_{1}p_{2}p_3p_4\Psi\right\rangle\right)
\nonumber\\&&-\frac{N!}{(N-4)!}\Im\left(\left\langle\Psi
,g^N_{2/7,\beta}(x_{1}-x_{2})
(\widehat{n}-2\widehat{n}_2+\widehat{n}_4)p_{1}p_{2}p_3p_4v^N_{\beta}(x_3-x_4)
\Psi\right\rangle\right)
\nonumber\\&=:&\sum_{j=0}^{13} T_j
\;.
\eea For $T_0$ to $T_{11}$ one can copy the estimates of $S_0$ to
$S_{10}$ above and gets, that $\sum_{j=0}^{11}T_j$ is bounded by
\be\label{bound}
C(\|\phi_t^{GP}\|_\infty+\|\nabla\phi_t^{GP}\|_\infty)\left(N^{-\gamma}+(\ln
N)^{1/3}\alpha(\Psi)\right)\;. \ee Instead of controlling $T_{12}$
and $T_{13}$ we add another term which pays respect to higher orders
of for the microscopic structure, i.e. we define \beas
\lambda_{3}(\Psi)&:=&\frac{N!}{(N-4)!}\Im\left(\left\langle\Psi
,g^N_{2/7,\beta}(x_{3}-x_{4}) g^N_{2/7,\beta}(x_{1}-x_{2})
(\widehat{n}-2\widehat{n}_2+\widehat{n}_4)p_{1}p_{2}p_3p_4\Psi\right\rangle\right)
\\&&-\frac{N!}{(N-4)!}\Im\left(\left\langle\Psi ,g^N_{2/7,\beta}(x_{1}-x_{2})
(\widehat{n}-2\widehat{n}_2+\widehat{n}_4)p_{1}p_{2}p_3p_4
g^N_{2/7,\beta}(x_{3}-x_{4}) \Psi\right\rangle\right) \eeas and the
respective $\lambda_{3}'(\Psi)$, again with
$\frac{d}{dt}\lambda_{3}(\Psi)=\lambda_{3}'(\Psi)$.

Controlling $\lambda_{3}'(\Psi)$ we get similar terms as the $T_j$
above, the only difference being an additional operator
$N^2g^N_{2/7,\beta}(x_j-x_k)p_jp_k$ and a higher order derivative of
$\widehat{n}$ (interpreting $\widehat{m}-\widehat{m}_1$ as the
derivative of $\widehat{m}$). We arrive at terms which are bounded
by (\ref{bound}) and the respective $T_{12}$ and $T_{13}$, i.e.
$$\frac{N!}{(N-6)!}\Im\left(\left\langle\Psi ,V_{5,6}g^N_{2/7,\beta}(x_{1}-x_{2}) g^N_{2/7,\beta}(x_{3}-x_{4})
(\widehat{n}-3\widehat{n}_2+3\widehat{n}_4-\widehat{n}_6)p_{1}p_{2}p_3p_4p_5p_6\Psi\right\rangle\right)\;,$$
$$\frac{N!}{(N-6)!}\Im\left(\left\langle\Psi ,g^N_{2/7,\beta}(x_{1}-x_{2}) g^N_{2/7,\beta}(x_{3}-x_{4})
(\widehat{n}-3\widehat{n}_2+3\widehat{n}_4-\widehat{n}_6)p_{1}p_{2}p_3p_4p_5p_6V_{5,6}\Psi\right\rangle\right)$$
and
$$\frac{N!}{(N-6)!}\Im\left(\left\langle\Psi ,V_{5,6}g^N_{2/7,\beta}(x_{1}-x_{2})
(\widehat{n}-3\widehat{n}_2+3\widehat{n}_4-\widehat{n}_6)p_{1}p_{2}p_3p_4p_5p_6g^N_{2/7,\beta}(x_{3}-x_{4})\Psi\right\rangle\right)\;.$$
Iteratively we add higher orders of the microscopic structure for
the remaining terms. Each iteration yields another operator
$N^2g^N_{2/7,\beta}(x_j-x_k)p_jp_k$ and a ``higher order derivative
of $\widehat{n}$'', thus a factor $N^{-1/7}$. We stop the iteration
as soon as all the remaining terms can be estimated by
(\ref{bound}). Thus we get
\begin{lemma}\label{replacealpa1}
There exists a $\gamma>0$ and functionals $\lambda_{1}(\Psi)$ and
$\lambda_{1}'(\Psi)$ such that
\begin{enumerate}
\item For any solution of the Sch\"odinger equation $\Psi_t\in\LZN$
$$i\frac{d}{dt}\lambda_{1}(\Psi_t)=\lambda_{1}'(\Psi_t)$$
\item
There exist a $C<\infty$ such that for any $\Psi\in\LZN$
$$|\lambda_{1}'(\Psi)-\alpha'_{1}(\Psi)|\leq C(\|\phi_t^{GP}\|_\infty+\|\nabla\phi_t^{GP}\|_\infty)\left(N^{-\gamma}+(\ln N)^{1/3}\alpha(\Psi)\right)\;.$$
\item $$\|\lambda_{1}(\Psi)\|\leq C N^{-\gamma}\|\phi_t^{GP}\|_\infty\;.$$
\end{enumerate}
\end{lemma}
Summarizing (\ref{estalpha}) , Lemma \ref{replacealpa2} and Lemma
\ref{replacealpa1} and setting
$\lambda(\Psi):=\alpha(\Psi)+\lambda_{1}(\Psi)+\lambda_{2}(\Psi)$ we
arrive at
\begin{corollary}\label{corollary}
There exists a $\gamma>0$ and functionals $\lambda(\Psi)$ and
$\lambda'(\Psi)$ such that
\begin{enumerate}
\item For any solution of the Sch\"odinger equation $\Psi_t\in\LZN$
$$i\frac{d}{dt}\lambda(\Psi_t)=\lambda'(\Psi_t)$$
\item
There exist a $C<\infty$ such that for any $\Psi\in\LZN$
$$|\lambda'(\Psi)|\leq C(\|\phi_t^{GP}\|_\infty+\|\nabla\phi_t^{GP}\|_\infty)\left(N^{-\gamma}+(\ln N)^{1/3}\lambda(\Psi)\right)\;.$$
\item $$\|\lambda(\Psi)-\alpha(\Psi)\|\leq C N^{-\gamma}\|\phi_t^{GP}\|_\infty\;.$$
\end{enumerate}

\end{corollary}

\subsection{Proof of Theorem \ref{theorem} for $\beta\geq1/3$}

In view of Corollary \ref{corollary} (c) and Lemma \ref{kondensat} (b) it suffices to prove that
$$\lim_{N\to\infty}\lambda(\Psi_t)=0$$ under the assumption $\lim_{N\to\infty}N^{\gamma}\lambda(\Psi_0)=0$. Therefore we use the estimates we get from Corollary \ref{corollary} (b)
on the time derivative of $\lambda(\Psi_t)$ and a Gronwall-like
argument.

Using that $$|\frac{d}{dt}\lambda(\Psi_t)|\leq
C(\|\phi_t^{GP}\|_\infty+\|\nabla\phi_t^{GP}\|_\infty)\left(N^{-\gamma}+(\ln
N)^{1/3}\lambda(\Psi_t)\right)$$ it follows that $\lambda(\Psi_t)$
is bounded from above by the solution $\mu_t$ of the differential
equation \be\label{difeqmu}\frac{d}{dt}\mu_t=
C(\|\phi_t^{GP}\|_\infty+\|\nabla\phi_t^{GP}\|_\infty)\left(N^{-\gamma}+(\ln
N)^{1/3}\mu_t\right)\ee with $\mu_{0}=\lambda(\Psi_0)$.

Defining $\zeta_t:=N^{-\gamma}+(\ln N)^{1/3}\mu_t$ we get from
(\ref{difeqmu}) \beas(\ln N)^{-1/3}\frac{d}{dt}\zeta_t=
C(\|\phi_t^{GP}\|_\infty+\|\nabla\phi_t^{GP}\|_\infty)\zeta_t\;.\eeas
Thus $$\zeta_t:=K \exp\left(C(\ln N)^{1/3}\int_0^t
(\|\phi_s^{GP}\|_\infty+\|\nabla\phi_s^{GP}\|_\infty)ds\right)$$
with $$K=\zeta_0=N^{-\gamma}+(\ln N)^{1/3}\mu_0=N^{-\gamma}+(\ln
N)^{1/3}\lambda(\Psi_0)<N^{-\gamma}(1+(\ln N)^{1/3})$$ for $N$ large
enough.

Note, that under the assumptions of the Theorem $\int_0^t (\|\phi_s^{GP}\|_\infty+\|\nabla\phi_s^{GP}\|_\infty)ds$ is bounded. Note also, that
$e^{C(\ln N)^{1/3}}=e^{C(\ln N)(\ln N)^{-2/3}}=N^{C(\ln N)^{-2/3}}$. Since $\lim_{N\to\infty}(\ln N)^{-2/3}=0$ it follows that $\lim_{N\to\infty}N^{-\gamma}e^{C(\ln N)^{1/3}}=0$ for any $\gamma>0$.

Thus $\zeta_t$ tends to zero as $N\to\infty$ uniform in $t< T$, so
does $\mu_t$ and so does $\lambda(\Psi_t)$. With Corollary
\ref{corollary} (c) the Theorem follows.

\section*{Acknowledgments} Helpful discussions with Detlef D\"urr,
Jakob Yngvason and Jean-Bernard Bru are gratefully acknowledged.

\section{Appendix}

It is left to prove the Lemma \ref{L2absch},  Lemma \ref{L2absch2}
and Lemma \ref{L2absch3}. Since $\|\phi^{GP}_t\|_\infty$ is bounded
we have that for any $m:\{1,\ldots,N\}\to\mathbb{R}^+$ with $m\leq
n^{-1}$
$$|\langle \Psi,q_1p_2|\phi^{GP}_t|^2\widehat{m}q_1q_2\Psi\rangle|\leq C\|q_1p_2\Psi\|\; \|\widehat{m}q_1q_2\Psi\|<C\alpha(\Psi)\;.$$

Note also that $p_jf(x_k)q_j=0$ for any $k\neq j$ and any function
$f$. So Lemma \ref{L2absch},  Lemma \ref{L2absch2} and Lemma
\ref{L2absch3} follow once we have \begin{lemma}\label{appendix} Let
$m:\{1,\ldots,N\}\to\mathbb{R}^+$ with $m\leq n^{-1}$, $0<\beta<1$.
Then we have under the conditions of the Theorem that there exists a
$C<\infty$ and a $\xi>0$  such that for any
$m:\{1,\ldots,N\}\to\mathbb{R}^+$ with $m\leq \sqrt{n}$

\begin{enumerate}

\item for any $0<\beta<1$
\bea \label{l21}&&\hspace{-1cm}|\langle\Psi ,p_1p_2
\left((N-1)v^N_{\beta}(x_1,x_2)-a|\phi^{GP}_t|^2(x_1)\right)p_1p_2
 \Psi\rangle|\\\nonumber&&\hspace{4cm}\leq C(\|\phi_t^{GP}\|_\infty+\|\nabla\phi_t^{GP}\|_\infty)N^{-\xi}
\\
\label{l22}&&\hspace{-1cm}|\langle\Psi ,p_1p_2
\left((N-1)v^N_{\beta}(x_1,x_2)-a|\phi^{GP}_t|^2(x_1)\right)
\widehat{m}q_1p_2
 \Psi\rangle|\\\nonumber&&\hspace{4cm}\leq C(\|\phi_t^{GP}\|_\infty+\|\nabla\phi_t^{GP}\|_\infty)N^{-\xi}
%
\\\label{l24} &&\hspace{-1cm}|\langle\Psi ,p_1 p_2 v^N_{\beta}(x_1,x_2) q_1 q_2
\Psi\rangle|\\&&\hspace{0.5cm}\nonumber\leq
CN^{-1}(\|\phi_t^{GP}\|_\infty+\|\nabla\phi_t^{GP}\|_\infty)(\alpha(\Psi)+N^{-\xi})
\\ \label{l25} &&\hspace{-1cm}|\langle\Psi ,q_1  p_2
v^N_{\beta}(x_1,x_2) \widehat{m}q_1 q_2 \Psi\rangle|
\\&&\hspace{0.5cm}\nonumber\leq
CN^{-1}(\|\phi_t^{GP}\|_\infty+\|\nabla\phi_t^{GP}\|_\infty)
(\alpha(\Psi)+N^{-\xi}+\|\mathds{1}_{\mathcal{S}_1}\nabla_1q_1\Psi\|^2)\eea
\item for any $0<\beta<1/3$
\bea \label{l23} &&\hspace{-1cm}|\langle\Psi ,p_1  p_2
v^N_{\beta}(x_1,x_2) \widehat{m}q_1 q_2
\Psi\rangle|\\&&\hspace{0.5cm}\nonumber \leq
CN^{-1}(\|\phi_t^{GP}\|_\infty+\|\nabla\phi_t^{GP}\|_\infty)
(\alpha(\Psi)+N^{-\xi})
\eea

\end{enumerate}
\end{lemma}

\begin{proof}

The right hand side of (\ref{l21}) is bounded by \beas
S_1&:=&\sup_{x_1\in\mathbb{R}^3}\left\{\left|\left\langle
\phi^{GP}_t(x_2),\left((N-1)v^N_{\beta}(x_1,x_2)-a|\phi^{GP}_t|^2(x_1)\right)\phi^{GP}_t(x_2)\right\rangle_2\right|\right\}%
\\&\leq&\sup_{x_1\in\mathbb{R}^3}\left\{\left|\left\langle
\phi^{GP}_t(x_1),(N-1)v^N_{\beta}(x_1-x_2)\phi^{GP}_t(x_1)\right\rangle_2-a|\phi^{GP}_t(x_1)|^2\right|\right\}
\\&&+(N-1)\sup_{|x_1-x_2|<CN^{-\beta}}\{(|\phi^{GP}_t(x_1))^2-(\phi^{GP}_t(x_2))^2|\}\|v^N_{\beta}\|_1\;.
 \eeas
The first term is equal to $((N-1)\|v^N_{\beta}\|_1-a)\|
\phi^{GP}_t(x_1)\|^2_\infty$ and in view of Definition \ref{defpot}
bounded by $C\|\phi^{GP}_t\|^2_\infty N^{-\delta}$. Using Taylors
formula the second term is of order $\|\nabla\phi^{GP}_t\|_\infty
N^{-\beta}$, thus
 \be\label{s1eq}|S_1|\leq
 C(\|\phi_t^{GP}\|_\infty+\|\nabla\phi_t^{GP}\|_\infty)(N^{-\beta}+N^{-\delta})\;.\ee
Since under our assumptions
$\|\phi_t^{GP}\|_\infty+\|\nabla\phi_t^{GP}\|_\infty<\infty$
(\ref{l21}) follows.

The left hand side of (\ref{l22}) is bounded by
$$S_1\|p_1p_2\Psi\|\;\|\widehat{m}q_1p_2\Psi\|\leq S_1\|p_1p_2\Psi\|\;\|\widehat{n}q_1p_2\Psi\|
$$
With (\ref{s1eq}) and Lemma \ref{kombinatorik} we get (\ref{l22}).

Next we shall prove (\ref{l23}). To estimate this term note, that
the operator norm of $p_1  p_2 v^N_{\beta}(x_1,x_2)  \widehat{m}q_1
q_2$ restricted to subspace of symmetric functions is much smaller
than the operator norm on full $\LZN$. Therefore one has to use
symmetry of $\Psi$ to get good control of this term. We define for
some $\delta>0$ we shall specify below the functions
$m^{a,b}:\{1,\ldots,N\}\to\mathbb{R}^+$ by $m^a(k):=m(k)$ for
$k<N^{1-\delta}$, $m^a(k)=0$ for $k\geq N^{1-\delta}$ and
$m^b=m-m^a$. It follows that (\ref{l23}) is bounded by
$$|\langle\Psi ,p_1  p_2
v^N_{\beta}(x_1,x_2) \widehat{m}^aq_1 q_2 \Psi\rangle|+|\langle\Psi
,p_1  p_2 v^N_{\beta}(x_1,x_2) \widehat{m}^bq_1 q_2 \Psi\rangle|_2
\Psi\rangle|\;.$$ Defining also $g:\{1,\ldots,N\}\to\mathbb{R}^+$ by
$g(k)=1$ for $k<N^{1-\delta}$, $g(k)=0$ for $k\geq N^{1-\delta}$ we
have that $m^a=m^as$ and thus

\bea&& \langle\Psi ,\widehat{g}_{-2}p_1 p_2 v^N_{\beta}(x_1,x_2) q_1
q_2 \widehat{m}^a\Psi\rangle
\nonumber\\&=&
(N-1)^{-1}\langle\Psi ,\sum_{j=2}^N\widehat{g}_{-2}p_1 p_j
v^N_{\beta}(x_1,x_j) q_1 q_j \widehat{m}^a\Psi\rangle
 \nonumber\\&\leq&
(N-1)^{-1}\|\sum_{j=2}^N \widehat{g}_{-2}q_jv^N_{\beta}(x_1,x_j)p_1
p_j\Psi\|\;\|\widehat{m}^a q_1 \Psi\|\;.
 \eea

 Using Lemma \ref{kombinatorik} (d) \bea&&
\langle\Psi ,p_1 p_2 v^N_{\beta}(x_1,x_2) q_1 q_2
\widehat{m}^b\Psi\rangle
\nonumber\\&=&
(N-1)^{-1}\langle\Psi ,\sum_{j=2}^N(\widehat{m}^b_{-2})^{1/2}p_1 p_j
v^N_{\beta}(x_1,x_j) q_1 q_j (\widehat{m}^b)^{1/2}\Psi\rangle
 \nonumber\\&\leq&
(N-1)^{-1}\|\sum_{j=2}^N q_jv^N_{\beta}(x_1,x_j)p_1
p_j\Psi\|\;\|(\widehat{m}^b)^{1/2} q_1 \Psi\|
\nonumber\\\label{l233}&\leq&
(N-1)^{-1}\|\sum_{j=2}^N q_jv^N_{\beta}(x_1,x_j)p_1
p_j\Psi\|\;\alpha(\Psi)\;.
 \eea
For any $h:\{1,\ldots,N\}\to\mathbb{R}^+$ we have that
\beas&&\|\sum_{j=2}^N \widehat{h}q_jv^N_{\beta}(x_1,x_j)p_1
p_j\Psi\|^2
\\&=&\sum_{j\neq k\neq 1}\langle \widehat{h}\Psi,p_1q_k\sqrt{v^N_{\beta}}(x_1,x_k)
p_j
\sqrt{v^N_{\beta}}(x_1,x_j)\\&&\hspace{2cm}\sqrt{v^N_{\beta}}(x_1,x_k)p_k\sqrt{v^N_{\beta}}(x_1,x_j)p_1
q_j\widehat{h}\Psi\rangle
\\&&+\sum_{j=2}^N\langle \widehat{h}\Psi,p_1  p_j
v_N(x_1,x_j)q_jv_N(x_1,x_j)p_1 p_j\widehat{h}\Psi\rangle
\\&\leq&(N-1)(N-2)\|\sqrt{v}_N(x_1,x_2)p_2\|_{op}^4\;\|
q_3\widehat{h}\Psi\|^2
\\&&+CN^{1/2}(N-1)\|(v^N_{\beta})^2\|_1\|\phi^{GP}_t\|_\infty^2\|\widehat{h}\|^2
\\&\leq&C(N-1)(N-2)N^{-2}\|\phi^{GP}_t\|_\infty^4\|
\widehat{h}\widehat{n}\Psi\|^2\\&&+C(N-1)N^{1/2}N^{-2+3\beta}\|\phi^{GP}_t\|_\infty^2\sup_{1\leq
k\leq N}|h(k)^2|
 \eeas
where we used Lemma \ref{kombinatorik} as well as that under our
conditions $\|v^N_{\beta}\|_\infty\leq C N^{3\beta}$.

Note that $\sup_{1\leq k\leq N}|g(k)^2|=1$ and $\sup_{1\leq k\leq
N}|m^b(k)|=N^{\delta}$. Note also that $\|
\widehat{m}^b_{-2}\widehat{n}\Psi\|^2\leq\|\widehat{n}^{1/2}_{-2}\Psi\|^2\leq\alpha(\Psi)+2N^{-1/2}$
and $s(k-2)n(k)<CN^{-\delta}$. Thus \beas |\langle\Psi ,p_1 p_2
v^N_{\beta}(x_1,x_2) \widehat{m}q_1 q_2 \Psi\rangle|\leq
CN^{-1}\|\phi^{GP}_t\|_\infty^2\left(\alpha(\Psi)+N^{-1+3\beta+2\delta}+
N^{-\delta}\right)\;.
 \eeas
Choosing $0<\delta<(-1+3\beta)/2$ and
$\xi<\min\{-1+3\beta+2\delta,\delta\}$ (\ref{l23}) follows.

 (\ref{l24}) for $0\leq \beta<1/3$ can be proven in the same way replacing $\widehat{m}$ by 1. For $1/3\leq \beta<1$ we define
$$U_N(\mathbf{x}):=\left\{
         \begin{array}{ll}
           \frac{3}{4\pi}\|v^N_{\beta}\|_1N^{3/4}, & \hbox{for $x<N^{-1/4}$;} \\
           0, & \hbox{else.}
         \end{array}
       \right.
$$
and \be \label{defh}h_N(x):=\int |x-y|^{-1}
(v^N_{\beta}(y)-U_N(y))d^3y \ee By this Definition it follows that
$h_N(x)=0$ for $x>N^{-1/4}$, $|h_N|<\|v^N_{\beta}\|_1 |x|^{-1}$,
$|\nabla h_N|<\|v^N_{\beta}\|_1|x|^{-2}$, thus
\be\label{hnorm}\|h_N\|_\infty<C
N^{-1+3\beta}\;\;\;\;\;\;\;\|h_N\|<CN^{-1-\beta/2}\ee and
$$-\Delta h_N=v^N_{\beta}-U_N\;.$$
So having proven (\ref{l24}) for $\beta<1/3$, (\ref{l24}) follows
once we have \be\label{l234}|\langle\Psi ,p_1  p_2 (\Delta
h_N)(x_1-x_2) q_1 q_2 \Psi\rangle| \leq
C(\|\phi_t^{GP}\|_\infty+\|\nabla\phi_t^{GP}\|_\infty) N^{-\xi}\;.
\ee Integration by parts and Lemma \ref{kombinatorik} (d) yield
\beas|\langle\Psi ,p_1 p_2 (\Delta h_N)q_1 q_2 \Psi\rangle| &\leq&
|\langle\Psi ,p_1  p_2  (\nabla_1 h_N(x_1-x_2))\nabla_1q_1 q_2
\Psi\rangle|
\\&&+|\langle\nabla_1p_1  p_2\Psi ,
(\nabla_1 h_N(x_1-x_2)) q_1 q_2 \Psi\rangle|
\\&=:&S_2+S_3 \;. \eeas
For $S_2$ we have similar as above \bea\label{mitnabla}
&&|\langle\Psi ,p_1 p_2  (\nabla_1 h_N(x_1-x_2))\nabla_1q_1 q_2
\Psi\rangle|\\&=&(N-1)^{-1}|\sum_{j=2}^N\langle\Psi ,p_1 p_j
(\nabla_1 h_N(x_1-x_j))\nabla_1q_1 q_j \Psi\rangle|
\\&\leq&(N-1)^{-1}\|\nabla_1q_1\Psi\|\;\|\sum_{j=2}^Nq_j
(\nabla_1 h_N(x_1-x_j)) p_1  p_j\Psi \| \eea For the last factor
we write \beas &&\|\sum_{j=2}^Nq_j (\nabla_1 h_N(x_1-x_j)) p_1
p_j\Psi \|\\&=&
\sum_{j\neq k\neq 1}\langle\Psi, p_1  p_k (\nabla_1
h_N(x_1-x_j))q_kq_j (\nabla_1 h_N(x_1-x_k)) p_1 p_j\Psi\rangle
\\&&+\sum_{j=2}^N\langle\Psi, p_1  p_j (\nabla_1
h_N(x_1-x_j))^2 p_1 p_j\Psi\rangle
\\&=:&S_4+S_5\;.
 \eeas
Note, that $\nabla_1 h_N(x_1-x_2)=-\nabla_2h_N(x_1-x_2)$, thus \beas
S_4&=&
\sum_{j\neq k\neq 1}\langle\Psi, p_1  p_kq_j (\nabla_j
h_N(x_1-x_j)) (\nabla_k h_N(x_1-x_k)) p_1 p_jq_k\Psi\rangle
 \eeas
Partial integrations yield \beas S_4&=&
\sum_{j\neq k\neq 1}\langle \nabla_j\nabla_k  p_1 p_kq_j\Psi
,h_N(x_1-x_j)h_N(x_1-x_k) p_1 p_jq_k\Psi\rangle
\\&&+\sum_{j\neq k\neq 1}\langle\nabla_j p_1  p_kq_j\Psi,
h_N(x_1-x_j)h_N(x_1-x_k)\nabla_k  p_1 p_jq_k\Psi\rangle
\\&&+\sum_{j\neq k\neq 1}\langle\nabla_k p_1  p_kq_j\Psi
,h_N(x_1-x_j) h_N(x_1-x_k) p_1 \nabla_jp_jq_k\Psi\rangle
\\&&+\sum_{j\neq k\neq 1}\langle  p_1  p_kq_j\Psi,
h_N(x_1-x_j)h_N(x_1-x_k)\nabla_j\nabla_k  p_1 p_jq_k\Psi\rangle\;,
 \eeas
so as above $S_4$ is bounded by the right hand side of (\ref{l24}).
For $S_5$ we estimate
$$\|(\nabla_1 h_N(x_1-x_j))^2\|_1=h_N(x_1-x_j)\Delta_1h_N(x_1-x_j)$$
which is (see below (\ref{defh})) of order $N^{-2+3\beta}$. Thus
$S_2$ is bounded by the right hand side of (\ref{l24}).

For $S_3$ note, that $\nabla_1 h_N(x_1-x_2)=-\nabla_2h_N(x_1-x_2)$.
Integration by parts yields \beas
S_3&\leq&|\langle\nabla_1\nabla_2p_1 p_2\Psi , h_N(x_1-x_2) q_1
q_2 \Psi\rangle|
\\&&+|\langle\nabla_1p_1 p_2\Psi , h_N(x_1-x_2) \nabla_2q_1
q_2 \Psi\rangle|
\\&\leq&\|\nabla\phi_t^{GP}\|_\infty^2\;\|h_N^2(x_1-x_2)\|_1^{1/2}\;\| q_1 q_2 \Psi\|
\\&&+\|\nabla\phi_t^{GP}\|_\infty\;\|\phi_t^{GP}\|_\infty\;\|h_N^2(x_1-x_2)\|_1^{1/2}\;\| \nabla_2q_1
q_2 \Psi\|
 \eeas
and (\ref{l234}) and thus (\ref{l24}) follows.

Next we shall prove (\ref{l25}). We define
 \be \label{defh2}h_N(x):=\int |x-y|^{-1} v^N_{\beta}(y)d^3y\;. \ee
As above this definition implies that
$|h_N|<\|v^N_{\beta}\|_1|x|^{-1}$, $|\nabla
h_N|<\|v^N_{\beta}\|_1|x|^{-2}$, $\|h_N\|_\infty<C N^{-1+3\beta}$,
$\|h_N\|<CN^{-1-\beta/2}$ and
$$-\Delta h_N=v^N_{\beta}\;.$$

\beas|\langle\Psi ,q_1  p_2 v^N_{\beta}(x_1-x_2) \widehat{m}q_1 q_2
\Psi\rangle|&=&|\langle\Psi ,\widehat{m}_2q_1  p_2 (\Delta
h_N)(x_1-x_2) q_1 q_2 \Psi\rangle|
\\&=&\langle\Psi ,q_1 p_2\widehat{m}_2 (\Delta h_N)q_1 q_2
\Psi\rangle| \\&\leq&
|\langle\Psi ,q_1  p_2\widehat{m}_2  (\nabla_1
h_N(x_1-x_2))\mathds{1}_{\mathcal{S}_1}\nabla_1q_1 q_2
\Psi\rangle|
\\&&+|\langle\Psi ,q_1  p_2\widehat{m}_2  (\nabla_1
h_N(x_1-x_2))\mathds{1}_{\overline{\mathcal{S}}_1}\nabla_1q_1 q_2
\Psi\rangle|
\\&&+|\langle\nabla_1q_1  p_2\widehat{m}_2\Psi ,
(\nabla_1 h_N(x_1-x_2)) q_1 q_2 \Psi\rangle|
\\&=:&S_6+S_7+S_8 \;. \eeas
For $S_6$ we have  \beas\nonumber
S_6&=&(N-1)^{-1}|\sum_{j=2}^N\langle\Psi ,q_1 p_j\widehat{m}_2
(\nabla_1 h_N(x_1-x_j))\mathds{1}_{\mathcal{S}_1}\nabla_1q_1 q_j
\Psi\rangle|
\\&\leq&(N-1)^{-1}\|\mathds{1}_{\mathcal{S}_1}\nabla_1q_1\Psi\|\;\|\sum_{j=2}^Nq_j
(\nabla_1 h_N(x_1-x_j)) \widehat{m}_2q_1  p_j\Psi \| \,.\eeas For
the last factor we write \beas &&\|\sum_{j=2}^Nq_j (\nabla_1
h_N(x_1-x_j)) \widehat{m}_2q_1 p_j\Psi \|^2\\&=&
\sum_{j\neq k\neq 1}\langle\Psi, \widehat{m}_2q_1  p_k (\nabla_1
h_N(x_1-x_j))q_kq_j (\nabla_1 h_N(x_1-x_j)) \widehat{m}_2q_1
p_j\Psi\rangle
\\&&+\sum_{j=2}^N\langle\Psi, \widehat{m}_2q_1  p_j (\nabla_1
h_N(x_1-x_j))^2 \widehat{m}_2q_1 p_j\Psi\rangle
\\&=:&S_9+S_{10}\;.
 \eeas
Note, that $\nabla_1 h_N(x_1-x_2)=-\nabla_2h_N(x_1-x_2)$, thus \beas
S_9&=&
\sum_{j\neq k\neq 1}\langle\Psi, \widehat{m}_2q_1  p_kq_j
(\nabla_j h_N(x_1-x_j)) (\nabla_k h_N(x_1-x_k))\widehat{m}_2q_1
p_jq_k\Psi\rangle
 \eeas
Partial integrations yield \beas S_9&=&
\sum_{j\neq k\neq 1}\langle \nabla_j\nabla_k  \widehat{m}_2q_1
p_kq_j\Psi ,h_N(x_1-x_j)h_N(x_1-x_k) \widehat{m}_2q_1
p_jq_k\Psi\rangle
\\&&+\sum_{j\neq k\neq 1}\langle\nabla_j \widehat{m}_2q_1  p_kq_j\Psi,
h_N(x_1-x_j)h_N(x_1-x_k)\nabla_k  \widehat{m}_2q_1
p_jq_k\Psi\rangle
\\&&+\sum_{j\neq k\neq 1}\langle\nabla_k \widehat{m}_2q_1  p_kq_j\Psi
,h_N(x_1-x_j) h_N(x_1-x_k) \widehat{m}_2q_1
\nabla_jp_jq_k\Psi\rangle
\\&&+\sum_{j\neq k\neq 1}\langle  \widehat{m}_2q_1  p_kq_j\Psi,
h_N(x_1-x_j)h_N(x_1-x_k)\nabla_j\nabla_k  \widehat{m}_2q_1
p_jq_k\Psi\rangle\;,
 \eeas
 Using symmetry of $\Psi$
\beas |S_9|&\leq&
2(N-1)(N-2)|\langle \mathds{1}_{\mathcal{S}_2}\nabla_2\nabla_3
\widehat{m}_2q_1 p_3q_2\Psi ,h_N(x_1-x_2)h_N(x_1-x_3)
\widehat{m}_2q_1 p_2q_3\Psi\rangle|
\\&&+2(N-1)(N-2)|\langle\mathds{1}_{\mathcal{S}_2}\nabla_2 \widehat{m}_2q_1  p_3q_2\Psi,
h_N(x_1-x_2)h_N(x_1-x_3)\mathds{1}_{\mathcal{S}_3}\nabla_3
\widehat{m}_2q_1 p_2q_3\Psi\rangle|
\\&&+2(N-1)(N-2)|\langle \mathds{1}_{\overline{\mathcal{S}}_2}\nabla_2\nabla_3
\widehat{m}_2q_1 p_3q_2\Psi ,h_N(x_1-x_2)h_N(x_1-x_3)
\widehat{m}_2q_1 p_2q_3\Psi\rangle|
\\&&+2(N-1)(N-2)|\langle\mathds{1}_{\overline{\mathcal{S}}_2}\nabla_2 \widehat{m}_2q_1  p_3q_2\Psi,
h_N(x_1-x_2)h_N(x_1-x_3)\mathds{1}_{\mathcal{S}_3}\nabla_3
\widehat{m}_2q_1 p_2q_3\Psi\rangle|
\\&&+2(N-1)(N-2)|\langle\mathds{1}_{\overline{\mathcal{S}}_2}\nabla_2 \widehat{m}_2q_1  p_3q_2\Psi,
h_N(x_1-x_2)h_N(x_1-x_3)\mathds{1}_{\overline{\mathcal{S}}_3}\nabla_3
\widehat{m}_2q_1 p_2q_3\Psi\rangle|
 \;,
 \eeas
so as above $S_9$ is bounded by the right hand side of (\ref{l24}).
For $S_{10}$ we estimate
$$\|(\nabla_1 h_N(x_1-x_j))^2\|_1=h_N(x_1-x_j)\Delta_1h_N(x_1-x_j)$$
which is (see below (\ref{defh})) of order $N^{-2+3\beta}$. Using
$|ab|< a^2+b^2$ we get that  $S_6$ is bounded by the right hand side
of (\ref{l25}).

For $S_7$ we have using symmetry 
\bea\label{ersterfaktor}
S_7&=&\frac{1}{N-1}|\sum_{j=2}^N\langle\Psi ,q_1  p_j  \widehat{m}_2(\nabla_1
h_N(x_1-x_j))\mathds{1}_{\overline{\mathcal{S}}_1}\nabla_1q_1 q_j
\Psi\rangle|
\\\nonumber&\leq&\frac{1}{N-1}\|\sum_{j=2}^Nq_j(\nabla_1
h_N(x_1-x_j))q_1  p_j
\widehat{m}_2\Psi\|\;\|\mathds{1}_{\overline{\mathcal{S}}_1}\nabla_1q_1
\Psi\| \;. \eea Using again symmetry we have
\bea\label{zweiterfaktor} &&\|\sum_{j=2}^Nq_j(\nabla_1
h_N(x_1-x_j))q_1  p_j \widehat{m}_2\Psi\|^2
\\\nonumber&\leq&\sum_{j=2}^N\|q_j(\nabla_1
h_N(x_1-x_j))q_1  p_j \widehat{m}_2\Psi\|^2
\\\nonumber&&+
\sum_{j\neq k\neq1}\langle\Psi,\widehat{m}_2q_1  p_k h_N(x_1-x_k))q_kq_j(\nabla_1
h_N(x_1-x_j))q_1  p_j \widehat{m}_2\Psi\rangle
\\\nonumber&=&(N-1)\|q_j(\nabla_1
h_N(x_1-x_2))q_1  p_2 \widehat{m}_2\Psi\|^2
\\\nonumber&&+
(N-1)(N-2)\langle\Psi,\widehat{m}_2q_1 q_3 p_2 (\nabla_1h_N(x_1-x_2))(\nabla_1
h_N(x_1-x_3))q_1 q_2 p_3 \widehat{m}_2\Psi\rangle
\eea
For the first summand we have
\bea\label{first}
(N-1)\|q_j(\nabla_1
h_N(x_1-x_2))q_1  p_2 \widehat{m}_2\Psi\|^2&\leq&(N-1)\|(\nabla_1
h_N(x_1-x_2))\|^2\;\|\phi^{GP}_t\|_\infty^2\;\|\widehat{m}_2q_1  p_2 \Psi\|
\nonumber\\&\leq&C(N-1)\|\phi^{GP}_t\|_\infty^2\;N^{-2}
N^{\beta}\leq C N^{\beta-1} \eea

For the second summand of the right hand side of
(\ref{zweiterfaktor}) we get using
$\nabla_1h_N(x_1-x_2)=-\nabla_2h_N(x_1-x_2)$, integrating by parts
and using symmetry \bea\label{verhau2}
&&(N-1)(N-2)\langle\Psi,\widehat{m}_2q_1 q_3 p_2
(\nabla_2h_N(x_1-x_2))(\nabla_3 h_N(x_1-x_3))q_1 q_2 p_3
\widehat{m}_2\Psi\rangle
\\\nonumber&=&2(N-1)(N-2)\langle \mathds{1}_{\mathcal{S}_3}\nabla_3 q_1 q_3 \nabla_2p_2\widehat{m}_2\Psi, (h_N(x_1-x_2))(
h_N(x_1-x_3))q_1 q_2 p_3 \widehat{m}_2\Psi\rangle
\\\nonumber&&+2(N-1)(N-2)\langle \nabla_3 q_1 q_3 \nabla_2p_2\widehat{m}_2\Psi, (h_N(x_1-x_2))(
h_N(x_1-x_3))\mathds{1}_{\overline{\mathcal{S}}_3}q_1 q_2 p_3
\widehat{m}_2\Psi\rangle
\\\nonumber&&+(N-1)(N-2)\langle q_1 q_3 \nabla_2p_2\widehat{m}_2\Psi, (h_N(x_1-x_2))(
h_N(x_1-x_3))q_1 q_2 \nabla_3p_3 \widehat{m}_2\Psi\rangle
\\\nonumber&&+(N-1)(N-2)\langle \mathds{1}_{\mathcal{S}_3}\nabla_3q_1 q_3 p_2\widehat{m}_2\Psi, (h_N(x_1-x_2))(
h_N(x_1-x_3))q_1 \mathds{1}_{\mathcal{S}_2}\nabla_2q_2 p_3
\widehat{m}_2\Psi\rangle
\\\nonumber&&+(N-1)(N-2)\langle \mathds{1}_{\mathcal{S}_3}\mathds{1}_{\overline{\mathcal{S}}_2}\nabla_3q_1 q_3 p_2\widehat{m}_2\Psi, (h_N(x_1-x_2))(
h_N(x_1-x_3)) q_1 \nabla_2q_2 p_3\widehat{m}_2\Psi\rangle
\\\nonumber&&+(N-1)(N-2)\langle \nabla_3q_1 q_3 p_2\widehat{m}_2 \Psi,(h_N(x_1-x_2))(
h_N(x_1-x_3)) \mathds{1}_{\overline{\mathcal{S}}_3}q_1 \nabla_2q_2
p_3\widehat{m}_2\Psi\rangle
\\\nonumber&\leq&2CN^2\|\mathds{1}_{\mathcal{S}_3}\nabla_3q_1
q_3\widehat{m}_2\Psi\|\;\|
\nabla\phi^{GP}_t\|_\infty\;\|h_N\|^2\;\|\phi^{GP}_t\|_\infty\;\|q_1
q_2 \widehat{m}_2\Psi\|
\\\nonumber&&+2N^2\|\nabla_3 q_1 q_3\widehat{m}_2\Psi\|\;\|
\nabla\phi^{GP}_t\|_\infty\;\|h_N\|\;
\|h_N\|_\infty\;\sqrt{N}\|\mathds{1}_{s_{1,2}}\|\;\|\phi^{GP}_t\|_\infty\;\|q_1
q_2 \widehat{m}_2\Psi\|
\\\nonumber&&+N^2\|q_1 q_2\widehat{m}_2
\Psi\|^2\;\|h_N\|^2\;\|\nabla\phi^{GP}_t\|_\infty^2
\\\nonumber&&+N^2\| \mathds{1}_{\mathcal{S}_2}\nabla_2q_1 q_2
p_2\widehat{m}_2\Psi\|^2\;\|h_N\|^2\;\|\phi^{GP}_t\|^2_\infty
\\\nonumber&&+2N^2\|\mathds{1}_{\overline{\mathcal{S}}_3}\|\;\|\phi^{GP}_t\|_\infty\;\|\nabla_3q_1 q_3 \widehat{m}_2\Psi\|
\;\|\phi^{GP}_t\|_\infty^2\;\|h_N\|^2\;
 \|q_1 \nabla_2q_2 p_3\widehat{m}_2\Psi\|\;.
 \eea Since with (\ref{1norms}) \beas \|\mathds{1}_{\overline{\mathcal{S}}_3}\|=\|\mathds{1}_{\overline{\mathcal{S}}_3}\|_1^{1/2}=\frac{3}{4\pi}N^{-17/9}\eeas

(\ref{verhau2}) is bounded by $C(\alpha(\Psi)+N^{-1/18})$. With
(\ref{ersterfaktor}), (\ref{zweiterfaktor}) and (\ref{first}) and
using $|ab|<a^2+b^2$ it follows that $S_7$ is bounded by the right
hand side of (\ref{l25}).

For $S_8$ note, that $\nabla_1 h_N(x_1-x_2)=-\nabla_2h_N(x_1-x_2)$.
Integration by parts yields \beas
S_8&\leq&|\langle\nabla_1\nabla_2q_1 p_2\Psi , h_N(x_1-x_2)
\widehat{m}_1q_1 q_2 \Psi\rangle|
\\&&+|\langle\nabla_1q_1 p_2\Psi , h_N(x_1-x_2) \nabla_2\widehat{m}_1q_1
q_2 \Psi\rangle|
\\&\leq&\|\nabla\phi_t^{GP}\|_\infty^2\;\|h_N\|\;\| \widehat{m}_1q_1 q_2 \Psi\|
\\&&+\|\nabla\phi_t^{GP}\|_\infty\;\|\phi_t^{GP}\|_\infty\;\|h_N\|\;\|\widehat{m}_1q_1 \nabla_2
q_2 \Psi\|
 \eeas
which is in view of Lemma \ref{kombinatorik} and (\ref{hnorm}) of
order $N^{-1-\beta/2}$ and (\ref{l25}) follows.

\end{proof}

\end{document}